\documentclass{llncs}

\usepackage{amssymb}
\usepackage{graphicx}
\usepackage{xcolor}
\usepackage{xspace}
\usepackage{amsmath}
\usepackage{stmaryrd}
\usepackage{hyperref}
\usepackage{times}
\usepackage{proof}
\usepackage{todonotes}
\usepackage{url}
\usepackage{tikz}
\usepackage{wrapfig}

\newif\iftr
\trtrue     

\newcommand{\lemmaApp}[2]{
\noindent
\textbf{Lemma #1.} \emph{#2}}

\newcommand{\propertyApp}[1]{
\noindent
\textbf{Property  #1.}}

\newcommand{\theoremApp}[1]{
\noindent
\textbf{Theorem  #1.}}

\newcommand{\pic}{$\pi$-calculus}
\newcommand{\respi}{$\mathit{ReS}\pi$}

\newcommand{\ifPi}{\texttt{if}}
\newcommand{\thenPi}{\texttt{then}}
\newcommand{\elsePi}{\texttt{else}}

\newcommand{\ce}[1]{\bar{#1}} 

\newcommand{\requestAct}[3]{\ce{#1}(#2) . #3}
\newcommand{\acceptAct}[3]{#1(#2) . #3}

\newcommand{\send}[2]{#1!\langle#2\rangle}
\newcommand{\sendAct}[3]{\send{#1}{#2} .  #3}
\newcommand{\receive}[2]{#1?(#2)}
\newcommand{\receiveAct}[3]{\receive{#1}{#2} . #3}

\newcommand{\select}[2]{#1 \triangleleft #2}
\newcommand{\selectAct}[3]{\select{#1}{#2} . #3}
\newcommand{\branching}[1]{#1 \triangleright}
\newcommand{\branch}[2]{#1\, :\, #2}
\newcommand{\branchSep}{,}
\newcommand{\branchAct}[2]{\branching{#1} \{#2\}}

\newcommand{\ifthenelseAct}[3]{\ifPi\ #1\ \thenPi\ #2\ \elsePi\ #3}

\newcommand{\inact}{\mathbf{0}}

\newcommand{\res}[1]{(\nu #1)}

\newcommand{\recAct}[2]{\mu #1.#2}

\newcommand{\alphaeq}{\equiv_{\alpha}}

\newcommand{\freev}[1]{\mathrm{fv}(#1)}
\newcommand{\freec}[1]{\mathrm{fc}(#1)}
\newcommand{\freese}[1]{\mathrm{fse}(#1)}

\newcommand{\ctrue}{\texttt{true}}
\newcommand{\cfalse}{\texttt{false}}

\newcommand{\mendpoint}[1]{[#1]}
\newcommand{\role}[1]{\mathtt{#1}}

\newcommand{\mrequestAct}[4]{\ce{#1}\mendpoint{#2}(#3) . #4}
\newcommand{\macceptAct}[4]{#1\mendpoint{#2}(#3) . #4}

\newcommand{\msend}[3]{#1\mendpoint{#2}!\langle#3\rangle}
\newcommand{\msendAct}[4]{\msend{#1}{#2}{#3} .  #4}
\newcommand{\mreceive}[3]{#1\mendpoint{#2}?(#3)}
\newcommand{\mreceiveAct}[4]{\mreceive{#1}{#2}{#3} . #4}

\newcommand{\mselect}[3]{#1\mendpoint{#2} \triangleleft #3}
\newcommand{\mselectAct}[4]{\mselect{#1}{#2}{#3} . #4}
\newcommand{\mbranching}[2]{#1\mendpoint{#2} \triangleright}
\newcommand{\mbranchAct}[3]{\mbranching{#1}{#2} \{#3\}}

\newcommand{\congr}{\equiv}
\newcommand{\subst}[2]{[ #1 / #2 ]}

\newcommand{\expreval}[2]{#1 \downarrow #2}
\newcommand{\red}{\rightarrow}
\newcommand{\rulelabel}[1]{[{\sc #1}]}

\newcommand{\fwredN}[1]{\twoheadrightarrow_{(#1)}}
\newcommand{\bwredN}[1]{\rightsquigarrow_{(#1)}}
\newcommand{\fwbwredN}[1]{\rightarrowtail_{(#1)}}

\newcommand{\boolType}{\mathsf{bool}}
\newcommand{\intType}{\mathsf{int}}
\newcommand{\sharedChanType}[1]{\langle#1\rangle}
\newcommand{\outType}[1]{![#1]}
\newcommand{\inpType}[1]{?[#1]}
\newcommand{\thrType}[1]{![#1]}
\newcommand{\catType}[1]{?[#1]}
\newcommand{\selType}[1]{\oplus[#1]}
\newcommand{\branchType}[1]{\&[#1]}
\newcommand{\inactType}{\mathsf{end}}
\newcommand{\recType}[1]{\mu #1}
\newcommand{\dual}[1]{\overline{#1}}

\newcommand{\sorting}{\Gamma}
\newcommand{\typing}{\Delta}
\newcommand{\basis}{\Theta}
\newcommand{\comp}{\cdot}
\newcommand{\judge}{\ \vdash\ }
\newcommand{\hasType}{\ \triangleright\ }

\newcommand{\singleSes}[2]{\langle #1:#2 \rangle \blacktriangleright} 
\newcommand{\stackComp}{\cdot}

\newcommand{\costBR}{\mathcal{C}_{br}}
\newcommand{\costMO}{\mathcal{C}_{mo}}

\makeatletter
\newcount\dirtree@lvl
\newcount\dirtree@plvl
\newcount\dirtree@clvl
\def\dirtree@growth{%
  \ifnum\tikznumberofcurrentchild=1\relax
  \global\advance\dirtree@plvl by 1
  \expandafter\xdef\csname dirtree@p@\the\dirtree@plvl\endcsname{\the\dirtree@lvl}
  \fi
  \global\advance\dirtree@lvl by 1\relax
  \dirtree@clvl=\dirtree@lvl
  \advance\dirtree@clvl by -\csname dirtree@p@\the\dirtree@plvl\endcsname
  \pgf@xa=.32cm\relax
  \pgf@ya=-.32cm\relax
  \pgf@ya=\dirtree@clvl\pgf@ya
  \pgftransformshift{\pgfqpoint{\the\pgf@xa}{\the\pgf@ya}}%
  \ifnum\tikznumberofcurrentchild=\tikznumberofchildren
  \global\advance\dirtree@plvl by -1
  \fi
}

\tikzset{
  dirtree/.style={
    growth function=\dirtree@growth,
    every node/.style={anchor=north},
    every child node/.style={anchor=west},
    edge from parent path={(\tikzparentnode\tikzparentanchor) |- (\tikzchildnode\tikzchildanchor)}
  }
}
\makeatother

\begin{document}

\pagestyle{plain} 

\mainmatter

\def\thetitle{Reversing Single Sessions}

\title{\thetitle
\thanks{This research has been partially founded by
EPSRC EP/K011715/1, EP/K034413/1 and EP/L00058X/1,
EU FP7 FETOpenX Upscale, 
MIUR PRIN Project CINA (2010LHT4KM),
and the COST Actions BETTY (IC1201) and Reversible computation (IC1405).} 
}

\titlerunning{\thetitle} 
\author{Francesco Tiezzi \and Nobuko Yoshida} 
\authorrunning{Tiezzi, Yoshida} 
\tocauthor{F. Tiezzi and N. Yoshida} 
\institute{University of Camerino, Italy\quad \email{francesco.tiezzi@unicam.it} 
\and
Imperial College London, UK\quad \email{n.yoshida@imperial.ac.uk}
\vspace*{-.45cm}}

\maketitle

\begin{abstract} 
Session-based communication has gained a widespread acceptance in practice
as a means for developing safe communicating systems via structured interactions.  
In this paper, we investigate how these structured interactions are affected by reversibility, which provides a 
computational model allowing executed interactions to be undone. In particular, we provide a systematic study 
of the integration of different notions of reversibility in both binary and multiparty single sessions. 
The considered forms of reversibility are: one for completely reversing a given session with one backward step, 
and another for also restoring any intermediate state of the session with either one backward step or multiple ones. 
We analyse the costs of reversing a session in all these different settings. 
Our results show that extending binary single sessions to multiparty ones does not affect the reversibility machinery and its costs.  
\end{abstract}

\section{Introduction}
\label{sec:intro}

In modern ICT systems, the role of communication is more and more crucial. 
This calls for a communication-centric programming style supporting safe and consistent 
composition of protocols. In this regard, in the last decade, primitives and related type theories supporting 
structured interactions, namely \emph{sessions}, among system participants have been 
extensively studied (see, e.g., \cite{HondaVK98,YoshidaV07,CoppoDY07,MostrousY09,HYCJACM15}). 

Another key aspect of ICT systems concerns their reliability. Recently, \emph{reversibility} 
has been put forward as a convenient support for programming reliable systems. 
In fact, it allows a system that has reached an undesired state to undo, in automatic fashion, 
previously performed actions. Again, foundational studies of mechanisms 
for reversing action executions have been carried out (see, e.g., 
\cite{DanosK04,DanosK05,DanosK07a,PhillipsU07,LaneseMS10,LaneseMSS11,CristescuKV13,RS}). 

In this paper, we investigate how the benefits of reversibility can be brought to 
structured communication and, hence, how reversibility and 
sessions affect each other. We concentrate on the primitives 
and mechanisms required to incorporate different notions of reversibility 
into two forms of session, and we analyse the costs of reversing a session
in these different settings. To study the interplay between reversibility and sessions 
we rely on a uniform foundational framework, based on \pic~\cite{PICALC}. 

Specifically, we focus on a simplified form of session, called \emph{single}, 
in which the parties that have created a session can only continue to interact along that 
single session. This setting permits to consider a simpler theoretical framework than 
the one with the usual notion of session, called here \emph{multiple}, where parties can 
interleave interactions performed along different sessions. This allows us to focus on the basic, 
key aspects of our investigation. Although single sessions are simpler, they 
are still  largely used in practice, and differently from multiple sessions (see, e.g., \cite{JLAMP_TY15,BarbaneraDd14}) 
their effect to reversibility is not studied yet in the literature. 

Concerning the parties involved in the sessions, we take into account both \emph{binary} and 
\emph{multiparty} sessions, which involve two or multiple parties, respectively. 
For each kind of session, we investigate the use of two forms of reversibility:
\emph{(i)} \emph{whole session} reversibility, where a single backward step 
reverses completely the given session, thus directly restoring its initialisation state; 
and \emph{(ii) session interactions} reversibility, where any intermediate state of 
the session can be restored, either in a \emph{(ii.a)~multi-step} or a 
\emph{(ii.b) single-step} fashion.
Figure~\ref{fig:cases} sums up the different combinations of sessions and 
reversibility we consider.

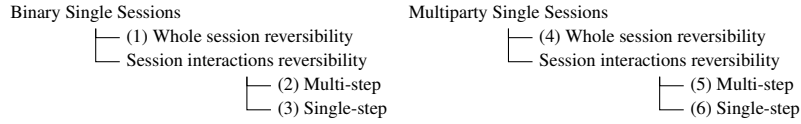
\begin{figure}[t]
\centering
\scriptsize
\begin{tikzpicture}[dirtree]
\node {Binary Single Sessions} 
            child { node {(1) Whole session reversibility} }
            child { node {Session interactions reversibility} 
            	child { node {(2) Multi-step} }
		child { node {(3) Single-step} }		
	};
\end{tikzpicture}
{\vspace*{-.4cm}}
\begin{tikzpicture}[dirtree]
\node {Multiparty Single Sessions} 
            child { node {(4) Whole session reversibility} }
            child { node {Session interactions reversibility} 
            	child { node {(5) Multi-step} }
		child { node {(6) Single-step} }		
	};
\end{tikzpicture}
\vspace*{-.1cm}
\caption{Sessions and Reversibility: considered combinations}
\label{fig:cases}
\vspace*{-.5cm}
\end{figure}

\begin{figure}[b]
\vspace*{-.7cm}
\centering
\includegraphics[scale=.3]{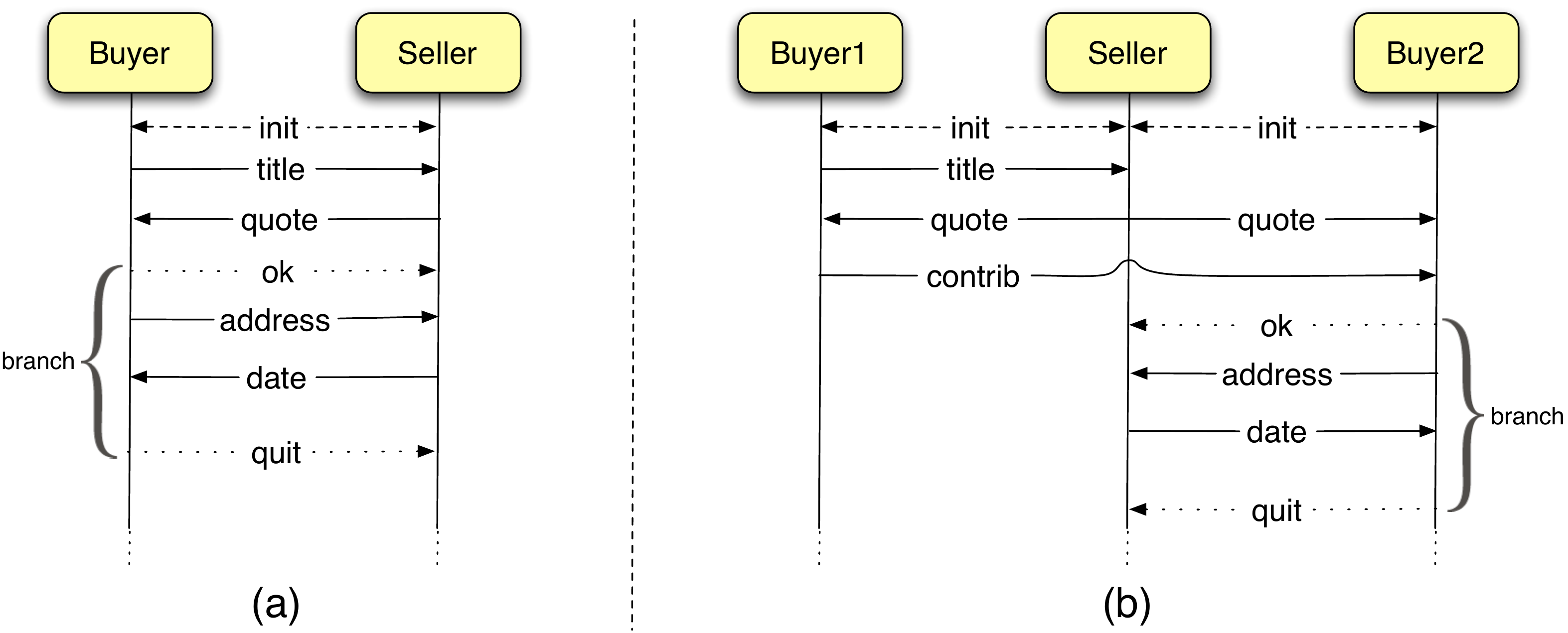}
\vspace*{-.5cm}
\caption{Single session protocols: Buyer-Seller (a) and Two-Buyers-Seller (b)}
\label{fig:examples}
\vspace*{-.5cm}
\end{figure}

We exemplify the reversible approaches throughout the paper by resorting to a typical business protocol example, 
drawn from \cite{HYCJACM15}. 
In case of binary session (Figure~\ref{fig:examples}.(a)), the protocol involves a Buyer willing to buy a book from a Seller.
Buyer sends the book title to Seller, which replies with a quote. 
If Buyer is satisfied by the quote, then he sends his address and Seller sends back the delivery date; 
otherwise Buyer quits the conversation.
In the multiparty case (Figure~\ref{fig:examples}.(b)), the above protocol is refined by considering 
two buyers, Buyer1 and Buyer2, that wish to buy an expensive book from Seller by combining 
their money. Buyer1 sends the title of the book to Seller, which sends the quote to both Buyer1 and Buyer2. 
Then, Buyer1 tells Buyer2 how much he is willing to contribute. Buyer2 evaluates how much 
he has to pay and either accepts, and exchanges the shipping information, or terminates the session.
In these scenarios, reversibility can be entered into the game to deal with 
errors that may occur during the interactions, or to make the protocols more 
flexible by enabling negotiation via re-iteration of some interactions. For example,
a buyer, rather than only accepting or rejecting a quote, can ask the seller for a new quote
by simply reverting the interaction where the current quote has been communicated. 
Similarly, Buyer2 can negotiate the division of the quote with Buyer1. Other possibilities allow the buyers to partially 
undo the current session, in order to take a different branch along the same session, 
or even start a new session with (possibly) another seller.

The contribution of this paper is twofold. Firstly, we show for each 
kind of session discussed above a suitable machinery that permits extending 
the corresponding non-reversible calculus in order to become reversible. 
Secondly, we compare the different cases, i.e. (1)-(6) in Figure~\ref{fig:cases}, 
by means of their costs for reverting a session, given in terms of 
number of backward steps and occupancy of the data structures used to 
store the computation history (which is a necessary information to reverse 
the effects of session interactions). 
Our results about reversibility costs are summarised in Figure~\ref{tab:costs}. 
\begin{wrapfigure}{r}{0.51\textwidth}
\small
\centering
\vspace*{-.7cm}
\begin{tabular}{l@{\ }c@{\ }|@{\ }c@{\ }|@{\ }c@{\ }|@{\ }c@{\ }|@{\ }c@{\ }|@{\ }c@{\ }|}
\cline{2-7}
&\multicolumn{3}{|@{\ }c@{\ }|}{Binary} &\multicolumn{3}{|@{\ }c@{\ }|}{Multiparty} \\
\cline{2-7}
\multicolumn{1}{@{\ }r|}{}
& $\ (1)$ & $(2)$ & $(3)$ & $(4)$ & $(5)$ & $(6)$ \\
\hline\\[-.28cm]\hline
\multicolumn{1}{|r|}{\#\,backward steps} & $1$ & $n$ & $1$ & $1$ & $n$ & $1$\\
\hline
\multicolumn{1}{|r|}{\#\,memory items} & $1$ & $n$ & $n$ & $1$ & $n$ & $n$  \\
\hline
\end{tabular}
\begin{tabular}{@{}l@{}}
{\scriptsize
\:$n$: number of interactions along the 
session to be reversed} 
\end{tabular}
\vspace*{-.2cm}
     \caption{Costs of Reversing Single Sessions}
     \label{tab:costs}
\vspace*{-.7cm}     
\end{wrapfigure}
It is worth noticing that linearity of sessions permits to achieve costs that are at most linear. 
Moreover, despite in case of complex interactions the multiparty approach provides a programming style more natural 
than the binary one, binary and multiparty sessions have the 
same reversibility costs. 
We discuss at the end of the paper how our work can extend to multiple sessions, which require 
a much heavier machinery for reversibility with respect to single ones, and have higher costs. 
This means that it is not convenient to use in the single session setting the same reversible 
machineries already developed for calculi with multiple sessions, which further motivates our 
investigation. 

The practical benefit of our systematic study is that it supplies a support to system designers 
for a conscious selection of the combination of session notion and reversibility mechanism 
that is best suited to their specific needs.

\emph{Summary of the rest of the paper.} 
Section~\ref{sec:host_calculus} provides background notions on binary and multiparty session-based 
variantes of \pic.
Section~\ref{sec:rev_single_sessions} shows how reversibility can be incorporated in single binary
sessions and what is its cost, while Section~\ref{sec:rev_single_sessions_multi} focusses on 
multiparty ones. 
Section~\ref{sec:conclusions} concludes by reviewing strictly related work and by touching upon directions for future work.
\iftr
For readers' convenience, further background material and proofs of results are collected in the Appendices.
\else
We refer to the companion technical report~\cite{RC_TR} for further background material and proofs of results.
\fi

\section{Background on session-based $\pi$-calculi}
\label{sec:host_calculus}

In this section, we give the basic definitions concerning two variants of the \pic, 
enriched with primitives for managing binary and multiparty sessions, respectively.

\medskip

\noindent
\textbf{Binary session calculus.}
The syntax definition of the binary session \pic\ \cite{YoshidaV07} relies on the following base sets: 
\emph{variables} (ranged over by $x$), storing values;
\emph{shared channels} (ranged over by $a$), used to initiate sessions;
\emph{session channels} (ranged over by $s$), consisting on pairs of \emph{endpoints} 
(ranged over by $s$, $\ce{s}$) used by the two parties to exchange values within an established session;
\emph{labels} (ranged over by $l$), used to select and offer branching choices;
and  \emph{process variables} (ranged over by $X$), used for recursion. 
Letter $u$ denotes \emph{shared identifiers}, i.e. shared channels and variables together; 
letter $k$ denotes \emph{session identifiers}, i.e. session endpoints and variables together;
letter $c$ denotes \emph{channels}, i.e. session channels and shared channels together. 
\emph{Values}, including booleans, integers, shared channels
and session endpoints, are ranged over by $v$.

\emph{Processes} (ranged over by $P$) and 
\emph{expressions} (ranged over by $e$ and defined by means of 
a generic expression operator $\text{op}$ representing standard operators on boolean and integer values) are 
given by the grammar in Figure~\ref{fig:syntax_pi}. 
\begin{figure}[t]
\centering
\small
	\begin{tabular}{@{}r@{\ }c@{\ }l@{\ \ }l@{}}
	$P$ & ::= & & \textbf{Processes} \\
	&             & $\requestAct{u}{x}{P}$ \ $\mid$ \ $ \acceptAct{u}{x}{P}$ 
	                   \ $\mid$ \
	                   $\sendAct{k}{e}{P}$ \ $\mid$ \ $\receiveAct{k}{x}{P}$ & \ request, accept ,output, input\\
	& $\mid$ & $\selectAct{k}{l}{P}$ \ $\mid$ \ $ \branchAct{k}{\branch{l_1}{P_1} \branchSep \ldots \branchSep \branch{l_n}{P_n}}$ 
	                   \ $\mid$ \
	                   $\inact$ \ $\mid$ \ $P_1 \!\mid\! P_2$ & \ selection, branching, inact, parallel\\
	& $\mid$ & $ \res{c}\, P$ \ $\mid$ \ $\ifthenelseAct{e}{P_1}{P_2}$ 
	                   \ $\mid$ \
	                   $X$ \ $\mid$ \ $\recAct{X}{P}$ & \ choice, restriction, recursion
	\\[.2cm]
	$e$ & ::= & \ \ $v$ \ $\mid$ \ $\text{op}(e_1,\ldots,e_n)$  & \textbf{Expressions}
	\\[.1cm]
	\hline
	\end{tabular}
	\vspace*{-.4cm}	
	\caption{Binary session calculus: syntax}
	\label{fig:syntax_pi}
\end{figure}

The operational semantics of the calculus is given in terms of a structural congruence and of a reduction relation, 
and is only defined for \emph{closed} terms, i.e. terms without free variables. 
The \emph{structural congruence}, written $\congr$, is standard 
\iftr
(cf. Figure~\ref{fig:congruence_pic} in Appendix).
\else
(see \cite{RC_TR}).
\fi
The \emph{reduction relation}, written $\red$, is the smallest relation on 
closed processes generated by the rules in Figure~\ref{fig:reduction_pic}.
We resort to the auxiliary function $\expreval{\cdot}{\,}$ 
for evaluating closed expressions: 
$\expreval{e}{v}$ says that expression $e$ evaluates to value $v$. 
Notationally, for $P$ a process, $\freev{P}$ denotes the set of free variables in $P$, 
and $\freese{P}$ the set of free session endpoints. 
\begin{figure*}[t]
	\centering
	\small
	\begin{tabular}{l@{\quad}l@{\ }l}
	$\requestAct{a}{x_1}{P_1}\ \mid \ \acceptAct{a}{x_2}{P_2}$
	$\ \red\ $
	$\res{s}(P_1\subst{\ce{s}}{x_1} \mid P_2\subst{s}{x_2})$ 
	& $s,\ce{s} \notin \freese{P_1,P_2}$ & \rulelabel{Con}
	\\[.2cm]
	$\sendAct{\ce{k}}{e}{P_1} \ \mid \ \receiveAct{k}{x}{P_2}$
	$\ \red\ $
	$P_1 \mid P_2\subst{v}{x}$
	& $(k=s\ \, \text{or}\ \, k=\ce{s}),\ \expreval{e}{v}$ & \rulelabel{Com}
	\\[.2cm]
	$\selectAct{\ce{k}}{l_i}{P} \ \mid \ \branchAct{k}{\branch{l_1}{P_1}  \branchSep \ldots \branchSep \branch{l_n}{P_n}}$
	$\ \red\ $
	$P \mid P_i$
	& $(k=s\ \, \text{or}\ \, k=\ce{s}),\ 1 \leq i \leq n$ & \rulelabel{Lab}
	\\[.2cm]
	$\ifthenelseAct{e}{P_1}{P_2}$
	$\ \red\ $
	$P_1$
	& $\expreval{e}{\ctrue}$ & \rulelabel{If1}
	\\[.2cm]
	$\ifthenelseAct{e}{P_1}{P_2}$
	$\ \red\ $
	$P_2$
	& $\expreval{e}{\cfalse}$ & \rulelabel{If2}
	\\[.1cm]
	\end{tabular}
	\begin{tabular}{c}
	$
	\infer[$\,\rulelabel{Par}$]{P_1 \!\mid\! P_2 \ \red\ P_1' \!\mid\! P_2}
	{P_1 \ \red\ P_1'}
	$	
	\qquad
	$
	\infer[$\,\rulelabel{Res}$]{\res{c}P \ \red\ \res{c}P'}
	{P \ \red\ P'}
	$
	\qquad
	$
	\infer[$\,\rulelabel{Str}$]{P_1 \ \red\ P_2}
	{P_1 \congr P_1'\ \red\ P_2' \congr P_2}
	$
	\\[.1cm]
	\hline
	\end{tabular}
	\vspace*{-.4cm}	
	\caption{Binary session calculus: reduction relation}
	\label{fig:reduction_pic}
	\vspace*{-.7cm}
\end{figure*}
We comment on salient points. 
A new session is established when two parallel processes synchronise 
via a shared channel $a$; this results on the generation of a fresh (private) session 
channel whose endpoints are assigned to the two session parties (rule \rulelabel{Con}).
Along a session, the two parties can exchange values (for data- and channel-passing, 
rule \rulelabel{Com}) and labels (for branching selection, rule \rulelabel{Lab}). 
The other rules are standard and state that: 
conditional choice evolves according to the evaluation of the expression argument 
(rules \rulelabel{If1} and \rulelabel{If2});
if a part of a larger process evolves, the whole process evolves accordingly 
(rules \rulelabel{Par} and \rulelabel{Res});
and structural congruent processes have the same reductions (rule \rulelabel{Str}). 

\label{sec:typingPic}

The syntax of \emph{sorts} (ranged over by $S$) and \emph{types} (ranged over by $\alpha$, 
$\beta$) used in the \emph{binary session type discipline} is defined in Figure~\ref{fig:typeSyntax}.
\begin{figure}[t]
	\centering
	\small
	\begin{tabular}{@{}r@{\ }c@{\ }l@{\ }l@{}}
	$S$ & ::= &  & \textbf{Sorts}\\
	&             & $\boolType$ \quad $\mid$ \quad $\intType$ \quad $\mid$ \quad $\sharedChanType{\alpha}$ & \ boolean, integer, shared channel
	\\[.2cm]
	$\alpha$ & ::= &  &  \textbf{Types}\\
	&             & $\outType{S}.\alpha$ \quad $\mid$ \quad $\thrType{\beta}.\alpha$ \quad $\mid$ \quad $\inpType{S}.\alpha$  
	\quad $\mid$ \quad $\catType{\beta}.\alpha$ & \ output, input\\
	& $\mid$ & $\selType{\branch{l_1}{\alpha_1}, \ldots, \branch{l_n}{\alpha_n}}$ 
	\ \  $\mid$ \ \ $\branchType{\branch{l_1}{\alpha_1}, \ldots, \branch{l_n}{\alpha_n}}$
	& \ selection, branching\\			
	& $\mid$ & $\inactType$ \quad $\mid$ \quad $t$ \quad $\mid$ \quad $\recType{t}.\alpha$ & \ end, recursion
	\\[.1cm]
	\hline
	\end{tabular}
	\vspace*{-.3cm}	
	\caption{Binary session calculus: sorts and types}
	\label{fig:typeSyntax}
	\vspace*{-.2cm}
\end{figure}
The type $\outType{S}.\alpha$ represents the behaviour of first outputting a value of sort $S$, 
then performing the actions prescribed by type $\alpha$; 
type $\thrType{\beta}.\alpha$ represents a similar behaviour, which starts with session output 
(\emph{delegation}) instead; 
types $\inpType{S}.\alpha$ and $\catType{\beta}.\alpha$ are the dual ones, 
receiving values instead of sending. 
Type $\selType{\branch{l_1}{\alpha_1}, \ldots, \branch{l_n}{\alpha_n}}$ represents the behaviour 
which would select one of $l_i$ and then behaves as $\alpha_i$, according to the selected $l_i$ 
(internal choice). 
Type $\branchType{\branch{l_1}{\alpha_1}, \ldots, \branch{l_n}{\alpha_n}}$ describes a branching 
behaviour: it waits with $n$ options, and behave as type $\alpha_i$ if the $i$-th action is selected 
(external choice).
Type $\inactType$ represents inaction, acting as the unit of sequential composition.
Type $\recType{t}.\alpha$ denotes a recursive behaviour, representing the behaviour that starts 
by doing $\alpha$ and, when variable $t$ is encountered, recurs to $\alpha$ again. 

\emph{Typing judgements} are of the form $\basis;\sorting \judge P \hasType \typing$,
where $\basis$, $\sorting$ and $\typing$, called \emph{basis}, \emph{sorting} and \emph{typing}
respectively,  are finite partial maps from shared identifiers to sorts, 
from session identifiers to types, and from process variables to typings, respectively
Intuitively,  the judgement $\basis;\sorting \judge P \hasType \typing$ 
stands for ``under the environment $\basis;\sorting$, process $P$ has typing $\typing$''. 
The axioms and rules defining the typing system are standard 
\iftr
(cf. Figure~\ref{fig:typingSysytem} in Appendix).
\else
(see \cite{RC_TR}).
\fi
 
\begin{example}[Buyer-Seller protocol]\label{ex:BS}
We show how the protocol in Figure~\ref{fig:examples}.(a) is rendered in the variant of \pic\
with binary sessions. 
The behaviour of Buyer is described by the following process:
$$
\small
\begin{array}{@{}r@{\ }c@{\ }l}
Buyer & \triangleq & 
\requestAct{a}{x}{\,}
\sendAct{x}{``\mathit{The\ Divine\ Comedy}"}{\,}
\receiveAct{x}{x_{quote}}{}\\
&& 
\ifPi\ \, x_{quote} \leq 20\ \
\thenPi\ \,
\selectAct{x}{l_{ok}}{\,} 
\sendAct{x}{addr()}{\,}
\receiveAct{x}{x_{date}}{\,P}
\ \ \elsePi\ \
\selectAct{x}{l_{quit}}{\,\inact} 
\end{array}
$$
This Buyer is interested in buying the Divine Comedy and is willing to pay not more than $20$ euros.  
The Seller participant instead is rendered as follows:
$$
\small
\begin{array}{@{}r@{\ }c@{\ }l}
Seller & \triangleq & 
\acceptAct{a}{z}{\,}
\receiveAct{z}{z_{title}}{\,}
\sendAct{z}{quote(z_{title})}{\,}
\branchAct{z}{\branch{l_{ok}\!\!}{\!\!
\receiveAct{z}{z_{addr}}{\,}
\sendAct{z}{date()}{Q\, }
} \branchSep 
\branch{l_{quit}\!\!}{\!\!\inact}}
\\
\end{array}
$$
Note that  $\mathit{addr()}$, $\mathit{quote()}$ and $\mathit{date()}$ are used to get a buyer address,
a quote for a given book, and the delivery date, respectively.
The overall specification is $Buyer \mid Seller$. 
\end{example}

\smallskip

\noindent
\textbf{Multiparty session calculus.}
The base sets for the synchronous multiparty session calculus \cite{KouzapasY14} are the same
of the binary case, except for \emph{session endpoints}, which now are denoted by 
$s\mendpoint{\role{p}}$, with $\role{p}$,$\role{q}$ ranging over \emph{roles} 
(represented as natural numbers). Thus, \emph{session identifiers} $k$ 
now range over session endpoints $s\mendpoint{\role{p}}$ or variables $x$.

The syntax of the calculus is defined by the grammar in Figure~\ref{fig:syntax_multipi},
where expressions $e$ are defined as in the binary case (with values that extends to 
multiparty session endpoints). 
\begin{figure}[t]
	\centering
	\small
	\begin{tabular}{@{}r@{\ }c@{\,}l@{\ }l@{}}
	$P$ & ::= &  & \textbf{Processes}\\
	&             & $\mrequestAct{u}{\role{p}}{x}{P}$ \ \ $\mid$ \ \ $\macceptAct{u}{\role{p}}{x}{P}$ 
	                   \ \ $\mid$ \ \
	                   $\msendAct{k}{\role{p}}{e}{P}$ \ \ $\mid$ \ \ $\mreceiveAct{k}{\role{p}}{x}{P}$ 
	               & \ request, accept, output, input\\
	& $\mid$ & $\mselectAct{k}{\role{p}}{l}{P}$ \ \ $\mid$ \ \ 
	                  $\mbranchAct{k}{\role{p}}{\branch{l_1\!\!}{\!\!P_1} \branchSep \ldots \branchSep \branch{l_n\!\!}{\!\!P_n}}$ 
	                  \ \ $\mid$ \ \ $\inact$ \ \ $\mid$ \ \ $P_1 \!\mid\! P_2$
	                  & \ selection, branching, inact, par.\\
	& $\mid$ &  $\ifthenelseAct{e}{P}{Q}$ \ \ $\mid$ \ \ $\res{c}\, P$
	                    \ \ $\mid$ \ \ 
	                   $X$ \ \ $\mid$ \ \  $\recAct{X}{P}$ 
	& \ choice, restriction, recursion
	\\[.1cm]
	\hline
	\end{tabular}
	\vspace*{-.4cm}	
	\caption{Multiparty session calculus: syntax}
	\label{fig:syntax_multipi}
	\vspace*{-.6cm}
\end{figure}
Primitive $\mrequestAct{u}{\role{p}}{x}{P}$ initiates a new session through 
identifier $u$ on the other multiple participants, each one of the form 
$\macceptAct{u}{\role{q}}{x}{P_{\role{q}}}$ where $1 \leq \role{q} \leq \role{p}-1$. 
Variable $x$ will be substituted with the session endpoint used for the interactions
inside the established session. 
Primitive $\msendAct{k}{\role{p}}{e}{P}$ denotes the intention of sending a value to 
role $\role{p}$; similarly, process $\mreceiveAct{k}{\role{p}}{x}{P}$ denotes the intention 
of receiving a value from role $\role{p}$. Selection and branching behave in a similar way. 

As usual the operational semantics is given in terms of a structural congruence and of a 
reduction relation. 
The rules defining the structural congruence are the same ones used for the binary calculus%
\iftr
~(Figure~\ref{fig:congruence_pic} in Appendix~\ref{app:background}), 
\else , 
\fi
where the rule for the scope extension of session channels takes into account the new form 
of session endpoints. 
The \emph{reduction relation} $\red$, instead, is the smallest relation on 
closed processes generated by the rules \rulelabel{If1}, \rulelabel{If2},
\rulelabel{Par}, \rulelabel{res} and \rulelabel{Str} in Figure~\ref{fig:reduction_pic},
and the additional rules in Figure~\ref{fig:reduction_pic_multi}.
\begin{figure*}[t]
	\centering
	\small
	\begin{tabular}{@{}l@{\hspace*{-1.3cm} }r@{\quad}r@{}}
	$\mrequestAct{a}{n}{x}{P_n}\, \mid \, \prod_{i=\{1,..,n-1\}}\macceptAct{a}{i}{x}{P_i}$ 
	$\ \ \red\ \ $
	&
	$s \notin \freese{P_i}$ with $i=\{1,..,n\}$
	&
	\rulelabel{M-Con}
	\\
	\qquad 	
	$\res{s}(P_n\subst{s\mendpoint{n}}{x} \mid \prod_{i=\{1,..,n-1\}}P_i\subst{s\mendpoint{i}}{x})$ 	
	\\[.3cm]
	$\msendAct{s\mendpoint{\role{p}}}{\role{q}}{e}{P}
	\, \mid \, \mreceiveAct{s\mendpoint{\role{q}}}{\role{p}}{x}{Q}$
	$\ \ \red\ \ $
	$P \mid Q\subst{v}{x}$
	& $\expreval{e}{v}$ & \rulelabel{M-Com}
	\\[.3cm]
	$\mselectAct{s\mendpoint{\role{p}}}{\role{q}}{l_i}{P}
	\, \mid \, \mbranchAct{s\mendpoint{\role{q}}}{\role{p}}{\branch{l_1}{P_1} \branchSep \ldots \branchSep \branch{l_n}{P_n}}$
	$\ \ \red\ \ $
	$P \mid P_i$
	& $1 \leq i \leq n$ & \rulelabel{M-Lab}
	\\[.1cm]
	\hline
	\end{tabular}
	\vspace*{-.4cm}	
	\caption{Multiparty session calculus: reduction relation (excerpt of rules)}
	\label{fig:reduction_pic_multi}
	\vspace*{-.5cm}	
\end{figure*}
We comment on salient points.
Rule \rulelabel{M-Con} synchronously initiates a session by requiring all session 
endpoints be present for a synchronous reduction, where each role $\role{p}$ 
creates a session endpoint $s\mendpoint{\role{p}}$ on a fresh session channel $s$. 
The participant with the maximum role ($\mrequestAct{a}{n}{x}{P_n}$) is 
responsible for requesting a session initiation. Rule \rulelabel{M-Com} defines how 
a party with role $\role{p}$ sends a value to the receiving party with role $\role{q}$. 
Selection and branching are defined in a similar way (rule \rulelabel{M-Lab}). 

The type discipline of this synchronous multiparty session calculus is simpler 
than the asynchronous one in \cite{HYCJACM15}, but it is much more elaborate than the binary 
case, as it considers global and local types. Therefore, due to lack of space, we 
relegate the definitions of types, as well as the rules of the corresponding type system,
to 
\iftr
Appendix~\ref{app:background}, 
\else
the companion technical report \cite{RC_TR}
\fi
and refer the interested reader to \cite{KouzapasY14} for a detailed account.

\begin{example}[Two-Buyers-Seller protocol]\label{ex_BBS}
We show how the protocol in Figure~\ref{fig:examples}.(b) is rendered in the variant of \pic\
with the multiparty sessions. 
The behaviour of Buyer1 and Buyer2 are described by the following processes:
$$
\small
\begin{array}{r@{\ }c@{\ }l}
Buyer1 & \triangleq & 
\mrequestAct{a}{3}{x}{\,}
\msendAct{x}{1}{``\mathit{The\ Divine\ Comedy}"}{\,}
\mreceiveAct{x}{1}{x_{quote}}{\,}
\msendAct{x}{2}{split(x_{quote})}{\,P_1}
\\[.2cm]
Buyer2 & \triangleq & 
\macceptAct{a}{2}{y}{\,}
\mreceiveAct{y}{1}{y_{quote}}{\,}
\mreceiveAct{y}{3}{y_{contrib}}{\,}
\ifPi\ \ y_{quote}-y_{contrib} \leq 10 \\
& & 
\thenPi\ \
\mselectAct{y}{1}{l_{ok}}{\,} 
\msendAct{y}{1}{addr()}{\,}
\mreceiveAct{y}{1}{y_{date}}{\,P_2}
\ \
\elsePi\ \
\mselectAct{y}{1}{l_{quit}}{\,\inact} 
\end{array}
$$
Now, Buyer1 divides the quote by means of the $split()$ function. 
The Seller process is similar to the binary case, but for the form of session 
endpoints:
$$
\small
\begin{array}{@{}r@{\ }c@{\ }l@{}}
Seller & \triangleq & 
\macceptAct{a}{1}{z}{\,}
\mreceiveAct{z}{3}{z_{title}}{\,}
\msendAct{z}{2}{quote(z_{title})}{\,}
\msendAct{z}{3}{\mathit{lastQuote}(z_{title})}{\,}\\
&& 
\mbranchAct{z}{2}{\branch{l_{ok}\!\!}{\!
\mreceiveAct{z}{2}{z_{addr}}{\,}
\msendAct{z}{2}{date()}{P_3\, }
} \branchSep\, 
\branch{l_{quit}\!\!}{\!\inact}}
\\
\end{array}
$$
where $\mathit{lastQuote}()$ simply returns the last quote computed for a 
given book. 
\end{example}

\section{Reversibility of single binary sessions}
\label{sec:rev_single_sessions}

This section formally introduces the notion of single session,
and illustrates constructs, mechanisms and costs to support the 
reversibility in the binary cases. 
\vspace*{-.2cm}

\subsection{Single sessions}
In the single sessions setting, when two processes start a session 
their continuations only interact along this single session. Thus, 
neither delegation (i.e., passing of session endpoints) 
nor initialisation of new sessions (also after the session closure)
is allowed.
As clarified below, the 
exclusive use of single sessions is imposed to processes by means of a 
specific type system, thus avoiding the use of syntactical constraints. 

Reversibility is incorporated in a process calculus typically by adding 
memory devices to store information about the computation history, 
which otherwise would be lost during computations.
In all single session cases, i.e. (1)-(6) in Figure~\ref{fig:cases}, 
we will extend the syntax of (binary/multiparty) session-based \pic\ as shown in 
Figure~\ref{fig:rev_single_sessions}.
\begin{figure}[t]
	\centering
	\small
	\begin{tabular}{@{}r@{\ \, }c@{\ \, }l@{\hspace{.3cm}}l@{}}
	\\[-1.2cm]
	$P$ & ::= & & \textbf{Reversible processes}\\
	&             & $\ldots$ \quad $\mid$ \quad $\singleSes{s}{m}P$ & \  \pic\ processes, single session
	\\[.1cm]
	$m$ & ::= & & \textbf{Memory stacks}\\
	&             & $P$ \quad $\mid$ \quad $P \stackComp m$& \ bottom element, push
	\\[.1cm]
	\hline
	\end{tabular}
	\vspace*{-.3cm}
	\caption{Reversible extension}
	\label{fig:rev_single_sessions}
	\vspace*{-.7cm}	
\end{figure}
The term \mbox{$\singleSes{s}{m}P$} represents a \emph{single session} along the channel $s$ 
with associated memory $m$ and code $P$. A \emph{memory} $m$ is a (non-empty) stack 
of processes, each one corresponding to a state of the session (the bottom element
corresponds to the term that initialised the session).
The term \mbox{$\singleSes{s}{m}P$} is a binder, i.e. it binds session channel $s$ in $P$.
In this respect, it acts similarly to operator $\res s P$, but the scope of \mbox{$\singleSes{s}{m}P$}
cannot be extended. 

In the obtained reversible calculi,  terms can perform, besides standard 
\emph{forward computations}, also  \emph{backward computations} 
that undo the effect of the former ones. 

We compare approaches (1)-(6) with respect to the \emph{cost} of 
reverting a session in the worst case, i.e. the cost of completely reverting a session. 
The cost is given in terms of \emph{(i)} the \emph{number of backward reductions} ($\costBR$), necessary to 
complete the rollback, and \emph{(ii)} \emph{memory occupancy} ($\costMO$), i.e. the number of element in the 
memory stack of the session when the rollback starts.
The two kinds of cost depend on the \emph{length} of the considered session, 
given by the number of (forward) steps performed along the session.

\subsection{Binary session reversibility}
\label{sec:single_binary_sessions}
Not all processes allowed by the syntax presented above corresponds to meaningful
processes in the reversible single sessions setting. 
Indeed, on the one hand, the syntax allows terms violating the single sessions 
limitation. 
On the other hand, in a general term of the calculus the history stored in its 
memories may not be consistent with the computation that has taken place.

We address the above issues by only considering a class of well-formed processes, called 
\emph{reachable} processes. In the definition of this class of processes, 
to ensure the use of single sessions only, as in \cite{HYCJACM15} we resort to 
the notion of \emph{simple} process. 
\begin{definition}[Simple process]\label{def:simple}
A process is \emph{simple} if (i)~it is generated by the grammar in Figure~\ref{fig:syntax_pi} and
(ii)~it is typable with a type derivation using prefix rules
where the session typings in the premise and the conclusion are restricted to at most a 
singleton\footnote{Using the standard typing system for the binary session \pic\ 
\iftr
(Figure~\ref{fig:typingSysytem} in Appendix~\ref{app:background}), 
\else
(see \cite[Figure~13]{RC_TR}), 
\fi
point (ii) boils down to:  $\typing$ of rules 
\rulelabel{Req}, 
\rulelabel{Acc},
\rulelabel{Send},
\rulelabel{Rcv},
\rulelabel{Sel}
and \rulelabel{Br} are empty; 
neither \rulelabel{Thr} nor \rulelabel{Cat} is used;
$\typing\comp \typing'$ in \rulelabel{Conc} contains at most a singleton; 
and $\typing$ of the remaining rules contain at most a singleton.}. 
\end{definition}
The point (i) of the above definition states that a simple process has no memory, while 
point (ii) states that each prefixed subterm in a simple process uses only a single session. 

The following properties clarify the notion of simple process.
\begin{property}\label{ex:noDelegation}
In a simple process, delegation is disallowed.
\vspace*{-.2cm}
\end{property}
\begin{proof}
The proof proceeds by contradiction and straightforwardly follows from Definition~\ref{def:simple}
\iftr
(see Appendix~\ref{app:proofsSingleSessions}).
\else
(see \cite{RC_TR}).
\fi
\end{proof}

\begin{property}\label{prop:subordinate}
In a simple process, subordinate sessions (i.e., new sessions initialised within the single session) are disallowed.
\vspace*{-.2cm}
\end{property}
\begin{proof}
The proof proceeds by contradiction 
\iftr
(see Appendix~\ref{app:proofsSingleSessions}).
\else
(see \cite{RC_TR}).
\fi
\end{proof}
We explain the meaning of subordinate session used in Property~\ref{prop:subordinate} 
by means of an example.  Let us consider the process $\acceptAct{a}{x}{\acceptAct{b}{y}{\sendAct{x}{1}P}}$
with \mbox{$y \in \freev{P}$}, 
which initialises a (subordinate) session using channel $b$ within a session previously initialised using channel 
$a$. This process is not simple, because typing it requires to type the (sub)process $\acceptAct{b}{y}{\sendAct{x}{1}P}$ 
under typing $x:\outType{\intType}.\alpha$, which is not empty as required by Definition~\ref{def:simple}.

Now, to ensure history consistency, as in~\cite{CristescuKV13} we only consider \emph{reachable} 
processes, i.e. processes obtained by means of forward reductions from simple processes. 

\begin{definition}[Reachable processes]\label{def:reachable}
The set of \emph{reachable} processes, for the case (i) in Figure~\ref{fig:cases} with $i \in \{1,2,3\}$, 
is the closure under relation $\fwredN{i}$ (see below)  
of the set of simple processes. 
\end{definition}
We clarify this notion by means of an example.
The process stored in the memory of term $\singleSes{s}{(\requestAct{a}{x}{\send{{x}}{1}{}}
\mid \acceptAct{a}{y}{\receive{y}{z}{}})}(\send{\ce{s}}{2}\mid\receive{s}{z})$ is not consistent with the related 
session process, because there is no way to generate the term $(\send{\ce{s}}{2}\mid\receive{s}{z})$ from 
the stored process (in fact, in the session process, the value sent along the endpoint $\ce{s}$ should be $1$
instead of $2$). 

We now present the semantics of the reversibility machinery in the three binary cases.
As usual, the operational semantics is given in terms of a structural congruence and 
a reduction relation\footnote{We use a reduction semantics with respect to a labelled one
because the former is simpler (e.g., it does not require to deal with scope extension 
of names) and, hence, is preferable when the labelled semantics is not needed 
(e.g., here we are not interested in labelled bisimulations). Moreover, 
works about session-based \pic\ use a reduction semantics, as well 
as many reversible calculi (e.g., \cite{LaneseMS10}, \cite{RS}, \cite{LienhardtLMS12}).}.
For all three cases, the laws defining the structural congruence 
are the same of the binary session calculus.
Instead, the \emph{reduction} relations for the cases (1)-(3), written $\fwbwredN{i}$ with $\textrm{i}\in \{1,2,3\}$, 
are given as the union of the corresponding \emph{forward reduction} relations $\fwredN{i}$ and \emph{backward reduction} relations $\bwredN{i}$,
which are defined by different sets of rules in the three cases. 
Notably, the rules in Figure~\ref{fig:reduction_binary}, whose meaning is straightforward,
are shared between the three cases. 

The semantics is only defined for closed, reachable terms, 
where now the definition of \emph{closed} term extends to session endpoints, 
in the sense that all occurrences of session endpoints $s$ and $\ce{s}$
have to be bound by a single session term $\singleSes{s}{m}\cdot$. 
This latter requirement is needed for ensuring that
every running session in the considered process can be reverted; 
for example, in the reachable process $(\send{\ce{s}}{1}\mid\receive{s}{x} \mid Q)$
there is a running session $s$ that cannot be reverted because no computation history 
information (i.e., no memory stack) is available for it. 
We discuss below the additional definitions for the operational semantics of the three binary cases.

\begin{figure*}[!t]
	\centering
	\small
	\begin{tabular}{c}
	$\ifthenelseAct{e}{P_1}{P_2} \ \fwredN{i}\  P_1\qquad \expreval{e}{\ctrue}$\qquad \rulelabel{Fw-If1}
	\\[.2cm]
	$\ifthenelseAct{e}{P_1}{P_2} \ \fwredN{i}\ P_2 \qquad \expreval{e}{\cfalse}$\quad\ \ \rulelabel{Fw-If2}
	\\[.2cm]
	$
	\infer[$\ \rulelabel{Fw-Par}$]{P_1 \mid \ P_2 \ \fwredN{i}\ P_1' \mid P_2}
	{P_1 \ \fwredN{i}\ P_1'}
	$	
	\qquad\qquad
	$
	\infer[$\ \rulelabel{Fw-Res}$]{\res{c}P \ \fwredN{i}\ \res{c}P'}
	{P \ \fwredN{i}\ P'}
	$
	\\[.2cm]
	$
	\infer[$\ \rulelabel{Fw-Str}$]{P \ \fwredN{i}\ Q}
	{P \congr P'\ \fwredN{i}\ Q' \congr Q}
	$	
	\\[.1cm]
	\hline
	\end{tabular}
	\vspace*{-.3cm}	
	\caption{Single sessions: shared forward and backward rules (for $i \in \{1..6\}$); 
	rules \rulelabel{Bw-Par}, \rulelabel{Bw-Res}, \rulelabel{Bw-Str} are omitted 
	(they are like the forward rules where $\bwredN{i}$ replaces $\fwredN{i}$)}
	\label{fig:reduction_binary}
	\vspace*{-.5cm}	
\end{figure*}

\paragraph{\textbf{(1) Whole session reversibility}.}
In this case the reversibility machinery of the calculus permits to undo only whole sessions. 
The forward rules additional to those in Figure~\ref{fig:reduction_binary}  are as follows:  
\begin{center}
\small
\begin{tabular}{@{}l@{\qquad}l@{\quad}r@{}}
$P \ \fwredN{1}\ \singleSes{s}{P}(P_1\subst{\ce{s}}{x} \mid P_2\subst{s}{y})$
&
$P=(\requestAct{a}{x}{P_1} \mid  \acceptAct{a}{y}{P_2})$
&\rulelabel{Fw(1)-Con}
\\[.3cm]
\multicolumn{3}{c}{
$
\infer[$\ \rulelabel{Fw(1)-Mem}$]{\singleSes{s}{m}P  \ \fwredN{1}\ \singleSes{s}{m}P'}
{P \ \red\ P'}
$
} 
\end{tabular}
\end{center}
Rule \rulelabel{Fw(1)-Con} initiates a single session $s$ with the initialisation term $P$ stored in the memory stack. 
As usual the two session endpoints $\ce{s}$ and ${s}$ replace the corresponding 
variables $x$ and $y$ in the two continuations $P_1$ and $P_2$ (within the scope of the 
single session construct). 
Notably, there is no need of using the restriction operator, because in the single 
session setting the session endpoints cannot be communicated outside the session, i.e., 
delegation is disallowed (Property~\ref{ex:noDelegation}). 
Moreover, differently from the non-reversible case (see rule \rulelabel{Con} in Figure~\ref{fig:reduction_pic}), 
there is also no need of requiring the session endpoints $s$ and $\ce{s}$ 
to be fresh in the session code (i.e., in $P_1$ and $P_2$), because 
of the notion of closed process given in this section. 
Rule \rulelabel{Fw(1)-Mem} simply states that a process within the scope of 
a single session evolves with a forward reduction according to 
its evolution with a standard reduction (defined in Figure~\ref{fig:reduction_pic}).

The only additional backward rule is the following one:
\begin{center}
\small
\begin{tabular}{l@{\qquad}r}
$\singleSes{s}{P}Q
\ \bwredN{1}\ 
P
$ 
&   \rulelabel{Bw(1)}
\end{tabular}
\end{center}
This rule permits to rollback the whole session conversation in every moment 
during its execution. In particular, the term that initialised the session is restored 
with a single backward reduction step. Notably, the fact that scope extension is
not allowed for the operator \mbox{$\singleSes{\cdot}{\cdot}\cdot$} ensures that  
the process $Q$ in \rulelabel{Bw(1)} does not contain processes not belonging 
to session $s$, i.e. unwanted deletions are prevented.
It is also worth noticing that the rollback of session $s$ does not involve other sessions,    
as no subordinate sessions can be active in $Q$ (Property~\ref{prop:subordinate}).

\paragraph{\textbf{(2) Multi-step.}}
In this case a session can be reversed either partially or totally. When the rollback starts, it 
proceeds step-by-step and can terminate in any intermediate state of the session, as well as 
in the initialisation state. The additional forward rules are:
\begin{center}
\small
\begin{tabular}{@{}l@{\qquad}l@{\quad}r@{}}
$P \ \fwredN{2}\ \singleSes{s}{P}(P_1\subst{\ce{s}}{x} \mid P_2\subst{s}{y})$
&
$P=(\requestAct{a}{x}{P_1} \mid  \acceptAct{a}{y}{P_2})$ 
&
\rulelabel{Fw(2)-Con}
\\[.3cm]
\multicolumn{3}{c}{
$
\infer[$\ \rulelabel{Fw(2)-Mem}$]{\singleSes{s}{m}P  \ \fwredN{2}\ \singleSes{s}{P \stackComp m}P'}
{P \ \red\ P'}
$
} 
\end{tabular}
\end{center}
Differently from the case (1), here it is necessary to keep track in the memory stack 
of each (forward) interaction that has taken place in the session. 
Therefore, the forward rule \rulelabel{Fw(2)-Mem} pushes the process $P$, representing the state before the 
transition, into the stack. 

The backward reduction relation is defined by the following additional rules:
\begin{center}
\small
\begin{tabular}{l@{\ \ \ }r@{\qquad\quad\ }l@{\ \ \ }r}
$\singleSes{s}{P}Q
\ \bwredN{2}\ 
P
$ 
&   \rulelabel{Bw(2)-1}
&
$\singleSes{s}{P \stackComp m}Q
\ \bwredN{2}\ 
\singleSes{s}{m}P
$ 
&   \rulelabel{Bw(2)-2}
\end{tabular}
\end{center}
Rule \rulelabel{Bw(2)-1} is like to  \rulelabel{Bw(1)}, but here it can be used  only when the memory stack 
contains just one element. The single session is removed because its initialisation 
state is restored. Rule \rulelabel{Bw(2)-2}, instead, permits undoing an intermediate state $Q$, by simply 
replacing it with the previous intermediate state $P$. In this case, since after the reduction the stack is not empty, the single session 
construct is not removed.

\paragraph{\textbf{(3) Single-step.}}
This is similar to the previous case, but the rollback (also of intermediate states) 
is always performed in a single step. The forward reduction relation is defined by the same rules of case
(2), where $\fwredN{3}$ replaces $\fwredN{2}$, while the backward one is defined by the following additional rules:
\begin{center}
\small
\begin{tabular}{@{}l@{\  }r@{\quad\ }l@{\  }r@{}}
$\singleSes{s}{P}Q
\ \bwredN{3}\ 
P
$ 
&   \rulelabel{Bw(3)-1}
&
$\singleSes{s}{P \stackComp m}Q
\ \bwredN{3}\ 
\singleSes{s}{m}P
$ 
&   \rulelabel{Bw(3)-2}
\\[.4cm]
$\singleSes{s}{m \stackComp  P}Q
\ \bwredN{3}\ P
$ 
&   \rulelabel{Bw(3)-3}
&
$\singleSes{s}{m' \stackComp  P \stackComp m}Q
\ \bwredN{3}\ 
\singleSes{s}{m}P
$ 
&   \rulelabel{Bw(3)-4}
\end{tabular}
\end{center}
Rules \rulelabel{Bw(3)-1} and \rulelabel{Bw(3)-2} are like \rulelabel{Bw(2)-1} and 
\rulelabel{Bw(2)-2}, respectively. In particular, rule \rulelabel{Bw(3)-2} is still used to 
replace the current state by the previous one. In addition to this, now rule \rulelabel{Bw(3)-4} permits 
to replace the current state $Q$ also by an intermediate state $P$ of the session computation; 
this is done in a single step. 
Notice that all interactions that took place after the one produced by  $P$ 
(i.e., the states stored in $m'$) are erased when $P$ is restored, while the previous interactions
(i.e., the states stored in $m$)  are kept. Notice also that the selection of the past state to 
restore is non-deterministic. A real-world reversible language instead should provide specific 
primitives and mechanisms to control reversibility (see discussion on Section~\ref{sec:conclusions}); 
anyway the controlled selection of the past states does  not affect the reversibility costs, hence this aspect 
is out-of-scope for this paper and left for future investigations.
Rule \rulelabel{Bw(3)-3} permits to directly undo the whole session from an 
intermediate state. 

\paragraph{\textbf{Results.}}
The cost of reverting a session in setting (1), in terms of both backward reductions ($\costBR$) and 
memory occupancy ($\costMO$), is \emph{constant} w.r.t. the session length.
In case (2), instead, the costs are linear in the length of the session
(recall that we consider the worst case, where the session is completely reversed).
Finally, in setting (3) the cost is constant in terms of backward 
reductions, and linear in terms of memory occupancy.
\begin{theorem}
\label{prop:costs}
Let $n$ be the length of a session, the costs of reverting it are:
case~(1) $\costBR=\costMO=1$; \ 
case~(2) $\costBR=\costMO=n$;\ 
case~(3) $\costBR=1$ and $\costMO=n$.\\[-.6cm]
\end{theorem}
\begin{proof}
The proof of case (1) is straightforward, while proofs of cases (2) and (3) proceed by induction on $n$
\iftr
(see Appendix~\ref{app:proofsSingleSessions}). 
\else
(see \cite{RC_TR}). 
\fi
\end{proof}

The following result shows that single binary sessions of cases (2) and (3)
enjoy a standard property of reversible calculi (Loop lemma, see \cite{DanosK04}): 
backward reductions are the inverse of the forward ones and vice versa.
\begin{lemma}[Loop lemma]\label{loopLemma}
Let $P = \singleSes{s}{m}Q$ and $P' = \singleSes{s}{m'}Q'$ be two reachable processes in setting (i), with $i\in \{2,3\}$.
$P \fwredN{i} P'$ if and only if $P' \bwredN{i} P$.
\end{lemma}
\begin{proof}
The proof for the \emph{if} (resp. \emph{only if}) part is by 
induction on the derivation of 
the forward (resp. backward) reduction 
\iftr
(see Appendix~\ref{app:proofsSingleSessions}). 
\else
(see \cite{RC_TR}). 
\fi
\end{proof}
Notably, case (1) does not enjoy this lemma because backward reductions do not allow to 
restore intermediate states of sessions. 

We conclude the section with an example showing the three approaches at work
on the Buyer-Seller protocol. 

\begin{example}[Reversible Buyer-Seller protocol]\label{ex:BS_rev}
Let us consider a reversible scenario concerning the Buyer-Seller protocol specified in Example~\ref{ex:BS}, where there are 
two sellers and a buyer. 

In case (1), the system evolves as follows:
$$
\small
\begin{array}{l}
Seller_1 \mid Seller_2 \mid Buyer
\\[.1cm]
\fwredN{1}\  Seller_1 \ \mid \ \singleSes{s}{Seller_2 \mid Buyer}(\receiveAct{s}{z_{title}}{P_s} \mid \sendAct{\ce{s}}{v_{title}}{P_b})
\\[.1cm]
\fwredN{1}^*\  Seller_1 \ \mid \ \singleSes{s}{Seller_2 \mid Buyer}(Q[\ldots] \mid P[\ldots,v_{date}/x_{date}]) \ \ = \ \ R
\end{array}
$$
where $v_{title}$ stands for $``\mathit{The\ Divine\ Comedy}"$. After these interactions between $Buyer$ and $Seller_2 $, 
the buyer has received a delivery date from the seller. In the unfortunate case that this date is not suitable for the buyer, 
the session can be reversed as follows:
$$
\small
\begin{array}{l}
R \ \bwredN{1}\  Seller_1 \mid Seller_2 \mid Buyer
\end{array}
$$
Now, $Buyer$ can start a new session with $Seller_2$ as well as with $Seller_1$. 

In case (2), the parties can reach the same state as follows:
$$
\small
\begin{array}{l}
Seller_1 \mid Seller_2 \mid Buyer
\
\fwredN{2}^*\  Seller_1 \ \mid \ \singleSes{s}{m}(Q[\ldots] \mid P[\ldots,v_{date}/x_{date}]) \ \, = \ \, R'
\end{array}
$$
where $m$ is $R_{date} \stackComp R_{addr} \stackComp R_{ok} \stackComp R_{if} \stackComp R_{quote} \stackComp R_{title}$, 
with $R_i$ denoting the process generating the interaction $i$. In this case, the buyer can undo only the last two session interactions 
as follows:
$$
\small
\begin{array}{l}
R' \ \bwredN{2}\  Seller_1 \mid \singleSes{s}{m'} R_{date} \ \bwredN{2}\  Seller_1 \mid \singleSes{s}{m''} R_{addr}
\end{array}
$$
with $m'=R_{addr} \stackComp m''$ and $m''=R_{ok} \stackComp R_{if} \stackComp R_{quote} \stackComp R_{title}$. Now, the buyer can 
possibly send a different address to the seller in order to get a more suitable date (as we assume $addr()$ and $date()$ be
two non-deterministic functions abstracting the interaction with buyer and seller backends). 

Finally, in case (3), the system can reach again the state $R'$, but this time the session can be also partially reversed by means of a single backward step:
$$
\small
\begin{array}{l}
R' \ \bwredN{3}\   Seller_1 \mid \singleSes{s}{m''} R_{addr}
\end{array}
$$
\end{example}

\section{Reversibility of single multiparty sessions}
\label{sec:rev_single_sessions_multi}

For the same motivations of the binary case, we do not consider all processes allowed by the syntax 
of the reversible multiparty single-session calculus, obtained by extending the grammar in 
Figure~\ref{fig:syntax_multipi} with the single sessions construct in Figure~\ref{fig:rev_single_sessions}. 
Again, we consider only reachable processes, whose definition relies on the notion of simple process. 

\begin{definition}[Multiparty simple process]\label{def:simple_multi}
A multiparty process is \emph{simple} if (i)~it is generated by the grammar in Figure~\ref{fig:syntax_multipi} and
(ii)~it is typable with a type derivation using the prefix rules where the session typings in the 
premise and the conclusion are restricted to at most a 
singleton\footnote{Using the typing system for the synchronous multiparty session \pic\ 
\iftr
(Figure~\ref{multi_type_sys} in Appendix~\ref{app:background}), 
\else
(see \cite[Figure~15]{RC_TR}), 
\fi
point (ii) boils down to:  $\typing$ of rules 
\rulelabel{MReq}, 
\rulelabel{MAcc},
\rulelabel{Send},
\rulelabel{Recv},
\rulelabel{Sel}
and \rulelabel{Bra} are empty; 
neither \rulelabel{Deleg} nor \rulelabel{Srecv} is used;
$\typing\comp \typing'$ in \rulelabel{Conc} contains at most a singleton; 
and $\typing$ of the remaining rules contain at most a singleton.}. 
\end{definition}

\begin{definition}[Multiparty reachable processes]\label{def:reachable_multi}
The set of \emph{reachable} processes, for the case (i) in Figure~\ref{fig:cases} with $i \in \{4,5,6\}$, 
is the closure under relation $\fwredN{i}$ (see below)  
of the set of multiparty simple processes. 
\end{definition}

We now present the semantics of the three multiparty cases.
For the definition of the reduction relations we still rely on the shared rules in 
Figure~\ref{fig:reduction_binary}. 

\paragraph{\textbf{(4) Whole session reversibility.}}
In case the reversibility machinery only permits to undo a whole session, 
the forward reduction relation is defined by these additional rules:  
\begin{center}
\small
\begin{tabular}{@{}r@{\hspace*{-2.8cm}}r@{}}
$P \ \fwredN{4}\ \singleSes{s}{P}(P_n\subst{s\mendpoint{n}}{x}
\mid \prod_{i=\{1,..,n-1\}}P_i\subst{s\mendpoint{i}}{x})
$
&
\rulelabel{Fw(4)-M-Con}
\\[.2cm]
&
$P=(\mrequestAct{a}{n}{x}{P_n}\, \mid \, \prod_{i=\{1,..,n-1\}}\macceptAct{a}{i}{x}{P_i})$ 
\\[.5cm]
\multicolumn{2}{c}{
$
\infer[$\ \rulelabel{Fw(4)-Mem}$]{\singleSes{s}{m}P  \ \fwredN{4}\ \singleSes{s}{m}P'}
{P \ \red\ P' }
$
} 
\end{tabular}
\end{center}
The meaning of these rules is similar to that of \rulelabel{Fw(1)-Con} and \rulelabel{Fw(1)-Mem}, with the only difference 
that the initialised single session is multiparty.  

The backward reduction relation is given by the same rules of case (1), of course defined for relation 
$\bwredN{4}$ instead of $\bwredN{1}$.

\paragraph{\textbf{(5) Multi-step.}} 
When sessions can be reversed also partially, in a step-by-step fashion, 
the additional forward rules are as follows:
\begin{center}
\small
\begin{tabular}{@{}r@{\hspace*{-2.8cm}}r@{}}
$P \ \fwredN{5}\ \singleSes{s}{P}(P_n\subst{s\mendpoint{n}}{x}
\mid \prod_{i=\{1,..,n-1\}}P_i\subst{s\mendpoint{i}}{x})
$
&
\rulelabel{Fw(5)-M-Con}
\\[.2cm]
&
$P=(\mrequestAct{a}{n}{x}{P_n}\, \mid \, \prod_{i=\{1,..,n-1\}}\macceptAct{a}{i}{x}{P_i})$ 
\\[.5cm]
\multicolumn{2}{c}{
$
\infer[$\ \rulelabel{Fw(5)-Mem}$]{\singleSes{s}{m}P  \ \fwredN{5}\ \singleSes{s}{P \stackComp m}P'}
{P \ \red\ P' }
$
} 
\end{tabular}
\end{center}
These rules are the natural extension of the corresponding rules of the binary version, i.e. case (2).
Their meaning, indeed, is the same. 

The backward reduction relation in setting (5) is given by the same rules of case~(2) where  
relation $\bwredN{2}$ is replaced by $\bwredN{5}$.

It is worth noticing that, by Definitions~\ref{def:simple_multi} and~\ref{def:reachable_multi}, 
concurrent interactions along the same session are prevented 
(by the type discipline in \cite{KouzapasY14}, which forces a linear use of session channels). 
Therefore, there is no need here to use a more complex reversible machinery (as in \cite{JLAMP_TY15}) for 
enabling a causal-consistent form of session reversibility.

\paragraph{\textbf{(6) Single-step.}}
As in the corresponding binary session setting, here the forward reduction relation is defined 
by the same rules of the multi-step case, i.e. case (5), where $\fwredN{6}$ replaces $\fwredN{5}$. 
Instead, the backward reduction relation is defined by the same rules of case (3), 
where $\bwredN{6}$ replaces $\bwredN{3}$. 

\paragraph{\textbf{Results.}}The cost of reverting a session in setting (4) is constant, 
while it is linear in the length of the session in case (5). In setting (6), instead,
the cost is constant in terms of backward reductions, and linear in terms of memory occupancy.
\begin{theorem}
\label{prop:costs_multi}
Let $n$ be the length of a session, the costs of reverting it are:
case~(4) $\costBR=\costMO=1$;\
case~(5) $\costBR=\costMO=n$;\
case~(6) $\costBR=1$ and $\costMO=n$.\\[-.6cm]
\end{theorem}
\begin{proof}
The proof of case (4) is a trivial adaptation of the proof of case (1) in Theorem~\ref{prop:costs}, 
while proofs of cases (5) and (6) proceed by induction on $n$
\iftr
(see Appendix~\ref{app:proofsSingleSessions_multi}). 
\else
(see \cite{RC_TR}). 
\fi
\end{proof}

In all multiparty approaches, cases (4)-(6), the backward computations have 
the same semantics of the corresponding binary approaches, cases (1)-(3), respectively.
An important consequence of this fact is that the cost of reverting a session in a 
multiparty case is the same of the corresponding binary case. In other words,
we can claim that \emph{extending binary sessions to 
multiparty ones, in the single session setting, does not affect the machinery for the  
reversibility and its costs.}

As in the binary case, also the multiparty sessions of cases (5) and (6) enjoy the Loop lemma.
\begin{lemma}[Loop lemma]\label{loopLemma_multi}
Let $P = \singleSes{s}{m}Q$ and $P' = \singleSes{s}{m'}Q'$ be two reachable processes in setting (i), with $i\in \{5,6\}$.
$P \fwredN{i} P'$ if and only if $P' \bwredN{i} P$.
\end{lemma}
\begin{proof}
The proof for the \emph{if} (resp. \emph{only if}) part is by 
induction on the derivation of 
the forward (resp. backward) reduction 
\iftr
(see Appendix~\ref{app:proofsSingleSessions_multi}). 
\else
(see \cite{RC_TR}). 
\fi
\end{proof}

We conclude by showing the three multiparty approaches at work on the Two-Buyers-Seller protocol. 
\begin{example}[Reversible Two-Buyers-Seller protocol]\label{ex:BBS_rev}
We consider a reversible scenario of the Two-Buyers-Seller session protocol specified in Example~\ref{ex_BBS}.

In case (4), the system can evolve as follows:
$$
\small
\begin{array}{l}
Buyer_1 \mid Buyer_2 \mid Seller
\\[.1cm]
\fwredN{4}\  \singleSes{s}{Buyer_1 \mid Buyer_2 \mid Seller}\\
\qquad\qquad
(
\msendAct{s[3]}{1}{v_{title}}{\,}
\mreceiveAct{s[3]}{1}{x_{quote}}{\,}P_{b1}
\ \mid \
\mreceiveAct{s[2]}{1}{y_{quote}}{\,}P_{b2}
\\
\qquad\qquad\ \, \mid \
\mreceiveAct{s[1]}{3}{z_{title}}{\,}
\msendAct{s[1]}{2}{quote(z_{title})}{\,}
\msendAct{s[1]}{3}{\mathit{lastQuote}(z_{title})}{\,}P_s
)
\\[.1cm]
\fwredN{4}^*\  \singleSes{s}{Buyer_1 \mid Buyer_2 \mid Seller}\\
\qquad\qquad
(P_{b1}[v_{quote}/x_{quote}] \mid P_{b2}[v_{quote}/y_{quote}] \mid P_{b1}[v_{title}/x_{title}] ) \ \ = \ \ R
\end{array}
$$
These interactions lead to a state where both buyers have received the seller's quote for the requested book. 
Now, if one of the two buyers is not satisfied with the proposed quote, he can immediately stop the session execution 
and reverse it with a single step:
$$
\small
\begin{array}{l}
R \ \bwredN{4}\  Buyer_1 \mid Buyer_2 \mid Seller
\end{array}
$$

In case (5), instead, the protocol execution can lead to a similar state, say $R'$, with the difference that 
the session memory $m$ keeps track of the traversed states, i.e. 
$m$ is $R_{quote2} \stackComp R_{quote1} \stackComp R_{title}$. 
In this case, the unsatisfied buyer can enact a sort of negotiation by undoing the last two session interactions:
$$
\small
\begin{array}{l}
R' \ \bwredN{5}\  \singleSes{s}{R_{quote1} \stackComp R_{title}} R_{quote2} 
\ \bwredN{5}\  \singleSes{s}{R_{title}} R_{quote1} 
\end{array}
$$
From the state $R_{quote1}$ the seller can compute again the quote for the requested book. 

Finally, in case (6), the system can reach again the state $R'$ and, likewise case (3), the session can be partially reversed 
by means of a single backward step:
$$
\small
\begin{array}{l}
R' \ \bwredN{6}\   \singleSes{s}{R_{title}} R_{quote1} 
\end{array}
$$
\end{example}

\section{Concluding remarks}
\label{sec:conclusions}

This work falls within a large body of research that aims at studying at foundational level the integration of reversibility 
in concurrent and distributed systems. In particular, reversible variants of well-established process calculi, such as CCS
and \pic, have been proposed as \emph{untyped} formalisms for studying reversibility mechanisms in these systems. Relevant works 
along this line of research have been surveyed in \cite{Lanese14}. Among them, we would like to mention the works that 
are closely related to ours, as they have been source of inspiration:  
RCCS \cite{DanosK04}, from which we borrow the use of memory stacks for keeping track of computation history;
$R\pi$ \cite{CristescuKV13}, from which we borrow the notion of reachable process (see Definition~\ref{def:reachable}); 
$R\mu Oz$ \cite{LienhardtLMS12}, which analyses reversibility costs in terms of space overhead;
and $\rho\pi$ \cite{LaneseMS10}, from which we borrow the use of a reversible reduction semantics (which is motivated by the fact that 
a labelled semantics would complicate our theoretical framework). 
However, all works mentioned above only focus on causal-consistent reversibility mechanisms for untyped concurrent systems,
without taking into account how they may impact on linearity-based structured interactions, which is indeed our aim. 
Moreover, none of the above work provides a systematic study of the different forms of reversibility we consider, namely whole session, 
multi-step and single-step, and of their costs.  

The works with the aim closest to ours are \cite{BarbaneraDd14} and \cite{JLAMP_TY15}.
The former paper studies session reversibility on a formalism based on session 
behaviours \cite{Barbanerad15}, which is a sort of sub-language of CCS with a checkpoint-based backtracking mechanism. 
The commonality with our work is the use of a one-size memory for each behaviour, which records indeed only the 
behaviour prefixed by the lastly traversed checkpoint. This resembles the one-size memories that we use in cases (1) 
and (4), with the difference that our checkpoints correspond to the initialisation states of sessions. On the other hand, 
session behaviours provide a formalism much simpler than session-based \pic, as e.g. message passing is not even considered.
Differently from our work, \cite{BarbaneraDd14} does not consider different solutions for enabling alternative forms of 
reversibility and does not provide a study of session reversibility complexity. 
The latter paper introduces \respi, a reversible variant of the \pic\ with binary multiple sessions. 
\respi\ embeds a multi-step form of reversibility and, 
rather than using a single stack memory per session, it uses a graph-like data structure and unique thread identifiers. 
Each element of this structure is devoted to record data concerning a single event, corresponding to the taking place of a communication 
action, a choice or a thread forking. Thread identifiers are used as links among memory elements, in order to form a structure 
for conveniently keeping track of the causal dependences among the session interactions. These dependences are crucial in the 
multiple session setting, where computations have to be undone in a \emph{causal-consistent} fashion \cite{DanosK04,Levy}, that is 
independent concurrent interactions can be undone in an order different from that of the forward computation. 
Differently from the present work, which considers both binary and multiparty session types, 
\cite{JLAMP_TY15} only focusses on the binary version. In addition, it does not address any cost issues
about reversible sessions. 

We plan to study the cost of reversing multiple sessions, where interactions 
along different sessions can be interleaved. 
By looking at \respi, we can see that passing from single sessions to multiple ones has significant impacts:
firstly, in terms of complexity of the memory structure, and secondly in terms of costs. In the multiple case,
reverting a session corresponds to revert a concurrent computation in a causal-consistent way, which requires 
to revert all interactions performed along other sessions that have a casual dependence with the interactions of the 
session to be reverted. This means that, in general, the cost is not defined only in terms of the length of the considered 
session, but it must include also the cost of reverting the depending interactions of other sessions. 

Another future direction that we plan to consider for our study concerns how the use of 
primitives and mechanisms to \emph{control} reversibility (see, e.g., \cite{LaneseMSS11}) affect our results. 
In a controlled approach for session reversibility, backward steps would not be always enabled, but they 
would be triggered by specific rollback actions. 

Moreover, we intend to extend our analysis on memory cost. In fact, the approach used in this work is 
\emph{coarse-grained}, as it is based on the number of elements of the stacks rather than on the amount 
of memory necessary for storing such elements. A fine-grained view is indeed  not necessary 
for the purpose of this work, as we just want to compare the different combinations 
of session and reversibility approaches and, in particular, to distinguish the whole session reversibility
with respect to the other cases, and we want to show that single-step interaction
reversibility and multi-step ones require memory with the same number of elements. 
Nevertheless, a fine-grained analysis on memory cost and an investigation of more compact 
representations of computation history would be interesting extensions of this work.

Finally, the enactment of reversibility is currently based on the information stored in the syntactical terms 
representing the involved processes. We plan to investigate the use of type information to enact and manage 
reversibility.

\bibliographystyle{plain} 
\bibliography{biblio}

\iftr 
\newpage
\appendix
\section{Background on session-based $\pi$-calculi}
\label{app:background}

In this Appendix, we report the definitions, omitted in Section~\ref{sec:host_calculus}, 
concerning the semantics and type discipline of the considered session-based variants of \pic.

\subsection{Binary session calculus}

\paragraph{Binders.}
Bindings are defined as follows:
$\requestAct{u}{x}{P}$, $\acceptAct{u}{x}{P}$ and $\receiveAct{k}{x}{P}$
bind variable $x$ in $P$;
$\res{a}\, P$ binds shared channel $a$ in $P$;
$\res{s}\, P$ binds session channel $s$ in $P$;
finally, $\recAct{X}{P}$ binds process variable $X$ in $P$.
The derived notions of bound and free names, alpha-equivalence $\alphaeq$, and 
substitution are standard. 
For $P$ a process, 
$\freev{P}$ denotes the set of \emph{free variables}, 
$\freec{P}$ denotes the set of \emph{free shared channels},
and $\freese{P}$ the set of \emph{free session endpoints}. 
%
For the sake of simplicity, we assume that free and bound variables are always chosen 
to be different, and that bound variables are pairwise distinct; the same applies to names.
Of course, these conditions are not restrictive and can always be fulfilled by possibly 
using alpha-conversion.

\paragraph{Structural congruence.}
The \emph{structural congruence}, written $\congr$, is defined as the smallest congruence relation 
on processes that includes the equational laws shown in Figure~\ref{fig:congruence_pic}. 
These are the standard laws of \pic: 
the first three are the monoid laws for $\mid$ (i.e., it is associative and commutative, and has 
$\inact$ as identity element); 
the second four laws deal with restriction and enable 
garbage-collection of channels, scope extension and scope swap, respectively; 
the eighth law permits unfolding a recursion (notation $P\subst{Q}{X}$ denotes 
replacement of free occurrences of $X$ in $P$ by process $Q$);
the last law equates alpha-equivalent processes, i.e.~processes only differing in the identity of 
bound variables/channels.

\begin{figure}[h]
\small
	\centering
	\begin{tabular}{l@{\qquad\ \ }l}	
	$(P \mid Q) \mid R \congr P \mid (Q \mid R)$ 
	&
	$P \mid Q \congr Q \mid P$ 
	\\[.3cm]
	$P \mid \inact \congr P$
	&
	$\res{c} \inact \congr \inact$
	\\[.3cm]
	$\res{a}P \mid Q \congr \res{a}(P \mid Q)$\ \ if $a \notin \freec{Q}$
	&
	$\res{c_1}\res{c_2}P \congr \res{c_2}\res{c_1}P$
	\\[.3cm]
	$\res{s}P \mid Q \congr \res{s}(P \mid Q)$\ \ if $s,\ce{s} \notin \freese{Q}$
	&
	$\recAct{X}{P} \congr P\subst{\recAct{X}{P}}{X}$
	\\[.3cm]
	$P \congr Q$\ \ if $P \alphaeq Q$
	\\[.2cm]
	\hline
	\end{tabular}
	\caption{Binary session calculus: structural congruence}
	\label{fig:congruence_pic}
\end{figure}

\paragraph{Types.}
As in \cite{YoshidaV07}, we take an \emph{equi-recursive} view of types, not distinguishing between 
a type $\recType{t}.\alpha$ and its unfolding $\alpha\subst{\recType{t}.\alpha}{t}$, and we are 
interested on \emph{contractive} types only, i.e. for each of sub-expressions 
$\recType{t}.\recType{t_1}\ldots\recType{t_n}.\alpha$  the body $\alpha$ is not $t$. 
Thus, in a typing derivation, types $\recType{t}.\alpha$ and $\alpha\subst{\recType{t}.\alpha}{t}$ can be used interchangeably. 

For each type $\alpha$, we define $\dual{\alpha}$, the \emph{dual type} of $\alpha$, 
by exchanging $!$ and $?$, and $\&$ and $\oplus$. The inductive definition is in Figure~\ref{fig:dualType}.
\begin{figure}[h]
	\centering
	\small
	\begin{tabular}{@{}r@{\ }c@{\ }l@{\quad}r@{\ }c@{\ }l@{\quad}r@{\ }c@{\ }l@{\quad}r@{\ }c@{\ }l@{}}
	$\dual{\outType{S}.\alpha}$ & = & $\inpType{S}.\dual{\alpha}$
	&
	$\dual{\outType{\beta}.\alpha}$ & = & $\inpType{\beta}.\dual{\alpha}$	
	&
	$\dual{\selType{\branch{l_i}{\alpha_i}}_{i\in I}}$ & = & $\branchType{\branch{l_i}{\dual{\alpha_i}}}_{i\in I}$
	&
	$\dual{\inactType}$ & = & $\inactType$
	\\[.3cm]
	$\dual{\inpType{S}.\alpha}$ & = & $\outType{S}.\dual{\alpha}$	
	&
	$\dual{\inpType{\beta}.\alpha}$ & = & $\outType{\beta}.\dual{\alpha}$	
	&
	$\dual{\branchType{\branch{l_i}{\alpha_i}}_{i\in I}}$ & = & $\selType{\branch{l_i}{\dual{\alpha_i}}}_{i\in I}$	
	&
	\multicolumn{3}{l}{$\dual{\recType{t}.\alpha} \,=\, \recType{t}.\dual{\alpha}
	\quad \dual{t}\,=\,t$}
	\\[.2cm]
	\hline
	\end{tabular}
	\caption{Dual types}
	\label{fig:dualType}
\end{figure}

\paragraph{Typing System.}
%
Typing judgement are of the form $\basis;\sorting \judge P \hasType \typing$.
The typing system is defined by the axioms and rules in Figure~\ref{fig:typingSysytem}. 
We call a typing \emph{completed} when it contains only $\inactType$ types. 
A typing $\typing$ is called \emph{balanced} if  whenever $s:\alpha,\ce{s}:\beta\in \typing$, then $\alpha=\dual{\beta}$.
We refer the interested reader to \cite{YoshidaV07} for detailed comments on the rules
and results.

\begin{figure}[p]
	\centering
	\small
	\begin{tabular}{@{}c@{}}
	$\sorting \judge \ctrue \hasType \boolType$\ \ \rulelabel{Bool$_{\mathit{tt}}$}
	\qquad
	$\sorting \judge \cfalse \hasType \boolType$\ \ \rulelabel{Bool$_{\mathit{ff}}$}
	\qquad
	$\sorting \judge 1 \hasType \intType$\ \ \rulelabel{Int}	
	\qquad \ldots
	\\[.4cm]
	$
	\infer[$\ \rulelabel{Sum}$]{\sorting \judge +(e_1,e_2) \hasType \intType}
	{\sorting \judge e_1 \hasType \intType & & \sorting \judge e_2 \hasType \intType}
	$		
	\qquad
	$
	\infer[$\ \rulelabel{And}$]{\sorting \judge \wedge(e_1,e_2) \hasType \boolType}
	{\sorting \judge e_1 \hasType \boolType & & \sorting \judge e_2 \hasType \boolType}
	$		
	\qquad \ldots
	\\[.2cm]
	\end{tabular}
	\begin{tabular}{@{\hspace*{-2cm}}r@{\hspace*{.4cm}}r}
	$\sorting\comp u:S \judge u \hasType S$\ \ \rulelabel{Id}	
	&
	\infer[$\ \rulelabel{Inact}$]{\basis;\sorting \judge \inact \hasType \typing}
	{\typing\ \  \text{completed}}	
	\\[.4cm]	
	$
	\infer[$\ \rulelabel{Req}$]{\basis;\sorting \judge \requestAct{u}{x}{P} \hasType \typing}
	{\sorting \judge u \hasType \sharedChanType{\alpha} & & \basis;\sorting \judge P \hasType \typing\comp x:\dual{\alpha}}
	$		
	&
	$
	\infer[$\ \rulelabel{Acc}$]{\basis;\sorting \judge \acceptAct{u}{x}{P} \hasType \typing}
	{\sorting \judge u \hasType \sharedChanType{\alpha} & & \basis;\sorting \judge P \hasType \typing\comp x:\alpha}
	$		
	\\[.4cm]	
	$
	\infer[$\ \rulelabel{Send}$]{\basis;\sorting \judge \sendAct{k}{e}{P} \hasType \typing\comp k:\outType{S}.\alpha}
	{\sorting \judge e \hasType S & & \basis;\sorting \judge P \hasType \typing\comp k:\alpha}
	$		
	&
	$
	\infer[$\ \rulelabel{Thr}$]{\basis;\sorting \judge \sendAct{k}{k'}{P} \hasType \typing\comp k:\thrType{\alpha}.\beta\comp k':\alpha}
	{\basis;\sorting \judge P \hasType \typing\comp k:\beta}
	$		
	\\[.4cm]	
	$
	\infer[$\ \rulelabel{Rcv}$]{\basis;\sorting \judge \receiveAct{k}{x}{P} \hasType \typing\comp k:\inpType{S}.\alpha}
	{\basis;\sorting\comp x:S \judge P \hasType \typing\comp k:\alpha}
	$		
	&
	$
	\infer[$\ \rulelabel{cat}$]{\basis;\sorting \judge \receiveAct{k}{x}{P} \hasType \typing\comp k:\catType{\alpha}.\beta}
	{\basis;\sorting \judge P \hasType \typing\comp k:\beta \comp x:\alpha}
	$		
	\\[.4cm]	
	\multicolumn{2}{c}{	
	$
	\infer[(1 \leq j \leq n)\  $\ \rulelabel{Sel}$]{\basis;\sorting \judge \selectAct{k}{l_j}{P} \hasType \typing
	\comp k:\selType{\branch{l_1}{\alpha_1}, \ldots, \branch{l_n}{\alpha_n}}}
	{\basis;\sorting  \judge P \hasType \typing \comp k:\alpha_j}
	$}
	\\[.4cm]	
	\multicolumn{2}{c}{	
	$
	\infer[$\ \rulelabel{Br}$]{\basis;\sorting \judge \branchAct{k}{\branch{l_1}{P_1} \branchSep \ldots \branchSep \branch{l_n}{P_n}} \hasType \typing
	\comp k:\branchType{\branch{l_1}{\alpha_1}, \ldots, \branch{l_n}{\alpha_n}}}
	{\basis;\sorting  \judge P_1 \hasType \typing \comp k:\alpha_1 && \ldots 
	 && \basis;\sorting  \judge P_n \hasType \typing \comp k:\alpha_n}
	$}
	\\[.4cm]		
	\multicolumn{2}{c}{
	$
	\infer[$\ \rulelabel{If}$]{\basis;\sorting \judge \ifthenelseAct{e}{P}{Q} \hasType \typing}
	{\sorting  \judge e \hasType \boolType &&
	\basis;\sorting  \judge P \hasType \typing && \basis;\sorting  \judge Q \hasType \typing}
	$}	
	\\[.4cm]	
	$
	\infer[$\ \rulelabel{Conc}$]{\basis;\sorting \judge P\mid Q \hasType \typing\comp \typing'}
	{\basis;\sorting  \judge P \hasType \typing && \basis;\sorting  \judge Q \hasType \typing'}
	$		
	&
	$
	\infer[$\ \rulelabel{Res1}$]{\basis;\sorting \judge \res{a}{P} \hasType \typing}
	{\basis;\sorting\comp a:S  \judge P \hasType \typing}
	$	
	\\[.4cm]	
	$
	\infer[$\ \rulelabel{Res2}$]{\basis;\sorting \judge \res{s}{P} \hasType \typing}
	{\basis;\sorting  \judge P \hasType \typing\comp s:\alpha \comp \ce{s}:\dual{\alpha}}
	$	
	&
	$
	\infer[$\ \rulelabel{Res3}$]{\basis;\sorting \judge \res{s}{P} \hasType \typing}
	{\basis;\sorting  \judge P \hasType \typing && s\ \text{ not in }\ \typing}
	$	
	\\[.4cm]	
	$\basis\comp X:\typing;\sorting \judge X \hasType \typing$\ \ \rulelabel{Var}			
	&
	\infer[$\ \rulelabel{Rec}$]{\basis;\sorting \judge \recAct{X}{P} \hasType \typing}
	{\basis\comp X:\typing;\sorting  \judge P \hasType \typing}	
	\\[.2cm]
	\hline
	\end{tabular}
	\caption{Typing System for binary session calculus}
	\label{fig:typingSysytem}
\end{figure}

\subsection{Multiparty session calculus}

\paragraph{Types.}
\emph{Global types}, ranged over by $G$, $G'$, \ldots describe the whole conversation scenario of a 
multiparty session as a type signature. The grammar of the global types is given in Figure~\ref{multi_types} (left).
\emph{Local types}, instead, correspond to the communication actions, representing sessions from the view-point of a single role.
Figure 5 (right) defines the syntax of local types.
\begin{figure}[!h]
\centering
\includegraphics[scale=.45]{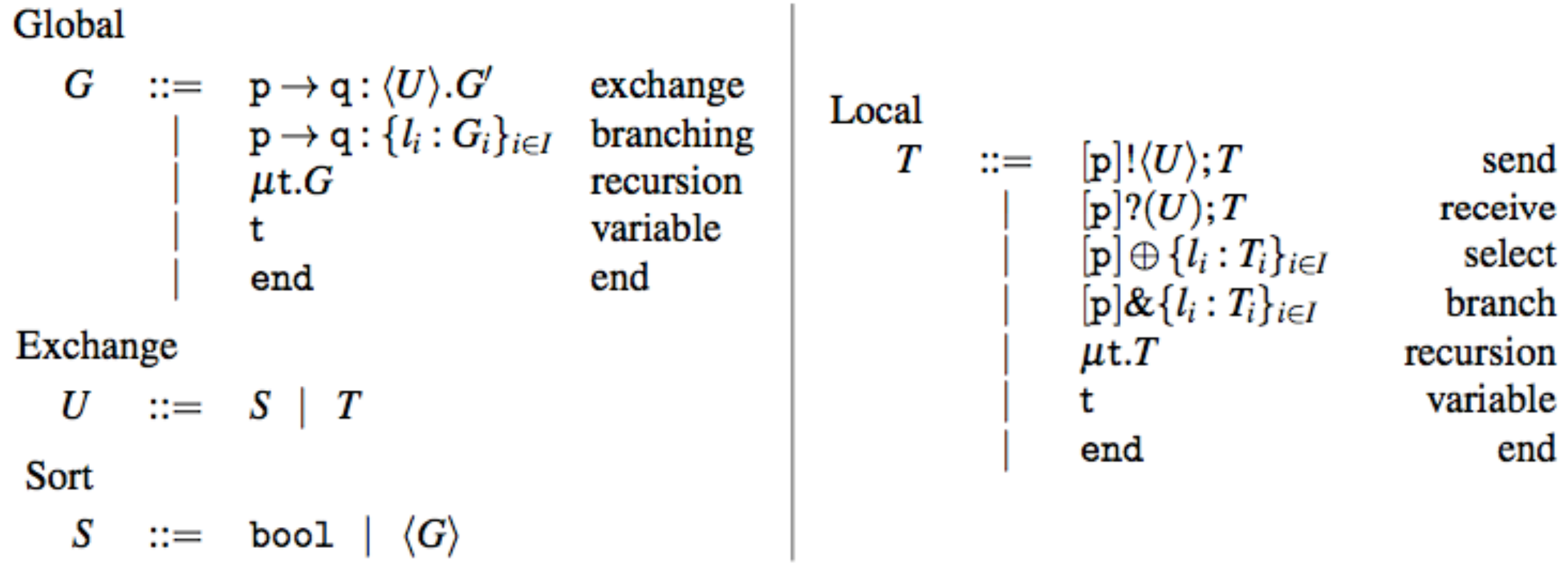}
\vspace*{-.2cm}
\caption{Global and local types}
\label{multi_types}
\vspace*{-.4cm}
\end{figure}

\paragraph{Typing System.}
Typing judgement are of the form $\sorting \judge P \hasType \typing$.
The typing system is defined by the axioms and rules in Figure~\ref{multi_type_sys}
(which is drawn from \cite{KouzapasY14} and can be reconciled with our notation by replacing symbols 
$c$, $\oplus$, $\&$, $\mathtt{tt}$ and $\mathtt{ff}$ by symbols $k$, $\triangleleft$, $\triangleright$, $\mathtt{true}$ and $\mathtt{false}$, respectively).
%
We refer the interested reader to \cite{KouzapasY14} for detailed comments on the rules
and results.
\begin{figure}[t]
\includegraphics[scale=.45]{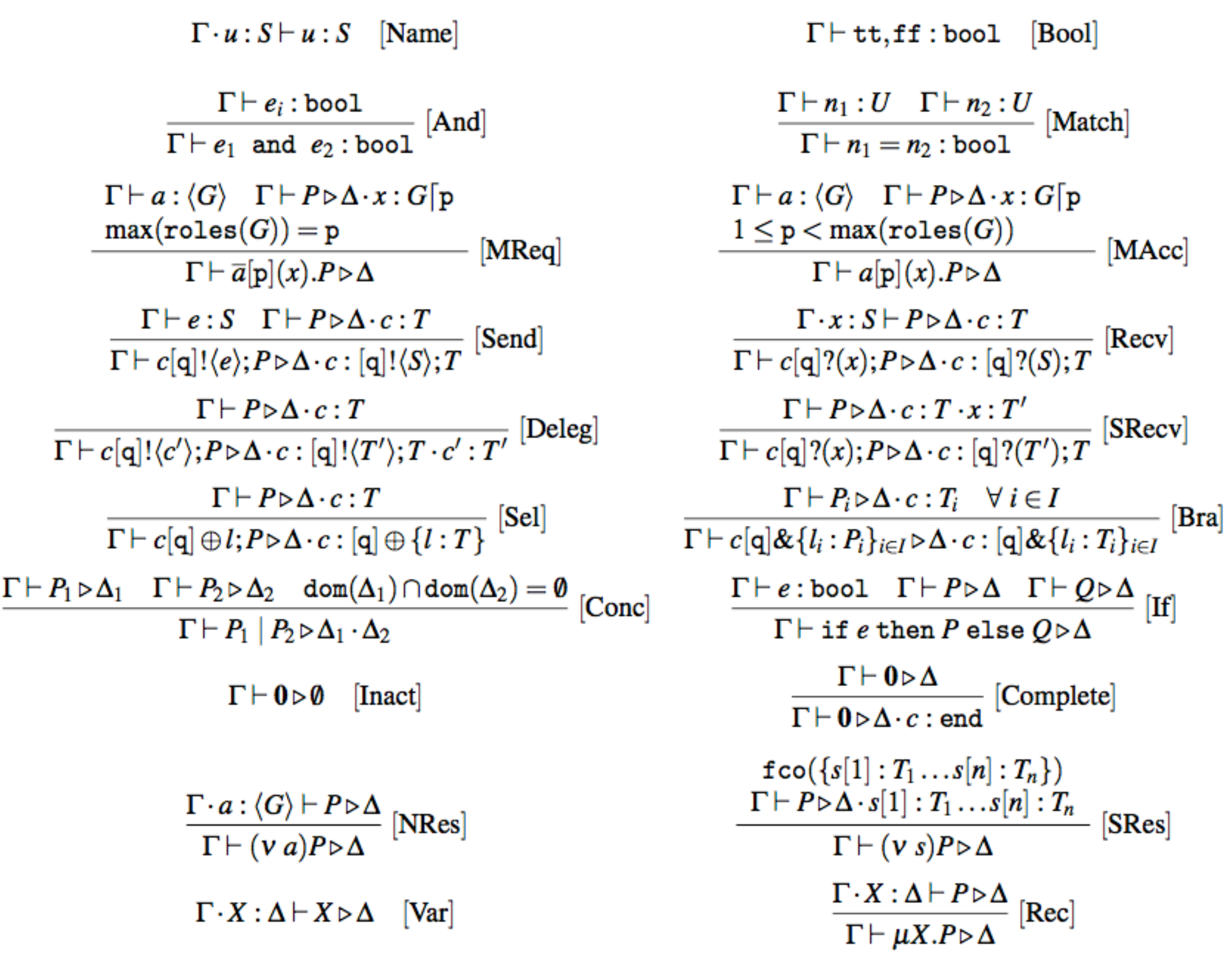}
\vspace*{-.8cm}
\caption{Typing system for synchronous multiparty session calculus}
\label{multi_type_sys}
\vspace*{-.4cm}
\end{figure}

\section{Proofs of Section~\ref{sec:rev_single_sessions}}
\label{app:proofsSingleSessions}

\propertyApp{\ref{ex:noDelegation}}
In a simple process, delegation is disallowed.
\begin{proof}
Delegation is achieved by passing session endpoints (which will be used, i.e. 
session endpoints not typed by end types) along a session. 
The proof proceeds by contradiction. Suppose that there exists a simple process 
that intends to perform a delegation by exchanging the session endpoint $s'$ 
on $\ce{s}$. Let us consider the sending term; it must have the form 
$\sendAct{\ce{s}}{s'}P$. However, this term is not typable under any $\basis$, $\sorting$ and 
typing $\typing$ of the form $\typing' \comp \ce{s}:\thrType{\alpha}.\beta\comp s':\alpha$.
In fact, due to the restrictions on the session typing required in Definiton~\ref{def:simple}, 
no rule to type session endpoint output (see rule \rulelabel{Thr} in Figure~\ref{fig:typingSysytem}) 
can be applied (because the typing in the rule conclusion cannot contain two sessions). 
Therefore, the whole process is not typable and, hence, 
according to Definition~\ref{def:simple} it cannot be simple, which contradicts the hypothesis. 
\qed
\end{proof}

\bigskip

\propertyApp{\ref{prop:subordinate}}
In a simple process, subordinate initialisation is disallowed.
\begin{proof}
We proceed by contradiction. Suppose that there exists a simple process 
that intends to start a subordinate session. We have two cases: the 
subordinate session is initialised either \emph{(i)} explicitly or \emph{(ii)} implicitly (using recursion). 
In case \emph{(i)}, the simple process contains a process, say $P$, that initialises a session 
using a given channel, say $b$, within a session previously initialised using a (possibly different) channel, say $a$. 
Let us assume, w.l.o.g., that the two initialisation actions are performed by $P$ one immediately after the other.  
Moreover let us consider the case where the outer session is used for sending a value (the other cases are similar).
Thus, process $P$ would have the form $\acceptAct{a}{x}{\acceptAct{b}{y}{\sendAct{x}{1}Q}}$
with $y \in \freec{Q}$. 
Now, by Definition~\ref{def:simple}, the judgement $\basis;\sorting \judge P \hasType \emptyset$ must hold for some $\basis$ 
and $\sorting$. To derive it, rule \rulelabel{Acc} in Figure~\ref{fig:typingSysytem} requires judgement $\basis;\sorting \judge \acceptAct{b}{y}{\sendAct{x}{1}Q} \hasType  x:\outType{\intType}.\alpha$ to hold. However, rule \rulelabel{Acc} cannot be applied again to derive this latter judgement, because 
typing $x:\outType{\intType}.\alpha$ is different from $\emptyset$ (as required by the conclusion of rule \rulelabel{Acc}). Thus, $P$ is not typable and, hence,
the whole process is not simple, which contradicts the hypothesis. 
Let us consider now the case \emph{(ii)}, where a subordinate session is created by resorting 
to recursion to reinitialise the session. We consider again the case where the session is used 
for sending a value (the other cases are similar). Thus, the simple process contains a process, say $R$, 
that has the form 
$\recAct{X}{\acceptAct{a}{y}{\sendAct{y}{1}}{X}}$, with $a$ having type
$\alpha = \outType{\intType}.{\inactType}$. However, process $R$ is not typable, as shown by the following derivation tentative:
$$
\small
\infer[$\ \rulelabel{Rec}$]{\emptyset;\sorting \judge R \hasType \emptyset}
{
\infer[$\ \rulelabel{Acc}$]{X:\emptyset;\sorting  \judge \acceptAct{a}{y}{\sendAct{y}{1}}{X} \hasType \emptyset}
{
\infer[$\ \rulelabel{Send}$]{X:\emptyset;\sorting \judge \sendAct{y}{1}{X} \hasType y:\outType{\intType}.{\inactType}}
{
X:\emptyset;\sorting \judge X \hasType y:\inactType\ \ $\ \rulelabel{Var}$  \lightning
}
}
}	
$$
where rule \rulelabel{Var} cannot be applied (denoted by symbol $\lightning$) because the type $\emptyset$ associated to $X$ by the basis 
differs from the typing $y:\inactType$.
Thus, the whole process is not simple, which contradicts the hypothesis. 
\qed
\end{proof}
Notably, even if recursion cannot be used in a simple process to reinitialise a session, inside a session recursion is still enabled. 
For example, $\recAct{X}{\sendAct{s}{1}}{X}$ is typable under typing $s:\recType{t}.\outType{\intType}.t$. Similarly, also parallel sessions 
are still allowed, like e.g. in $(\requestAct{a}{x_1}{P_1} \mid \acceptAct{a}{x_2}{P_2} \mid \requestAct{b}{x_3}{P_3} \mid \acceptAct{b}{x_4}{P_4})$ 
or in $\recAct{X}{(\requestAct{a}{x_1}{P_1} \mid \acceptAct{a}{x_2}{P_2} \mid X)}$.

\bigskip

\theoremApp{\ref{prop:costs}}
Let $n$ be the length of a session, the cost of reverting it is\\ 
\mbox{\ \ } -- case~(1): $\costBR=\costMO=1$;\\
\mbox{\ \ } -- case~(2): $\costBR=\costMO=n$;\\
\mbox{\ \ } -- case~(3): $\costBR=1$ and $\costMO=n$.
\begin{proof} We proceed separately with the three cases.
\begin{itemize}
\item \textbf{Case (1).}\quad
The proof is straightforward. Let $\singleSes{s}{m}Q$ be the term corresponding to 
a given session to be reverted. Concerning the occupancy of the memory stack $m$, 
exactly one element is stored in $m$ before performing the rollback. 
Indeed, according to rule \rulelabel{Fw(1)-Con}, when a session is initiated, its memory $m$ 
only contains the initiating process, say $P$. The single session term then can evolve only 
by means of application of rule \rulelabel{Fw(1)-Mem}, which never modifies the memory $m$. Therefore, 
the session always keeps the form $\singleSes{s}{P}\cdot$, i.e. its memory always contains just one element
($\costMO=1$). 
Concerning the cost in terms of backward reduction steps, by applying rule \rulelabel{Bw(1)}, the session
$\singleSes{s}{P}Q$ can be completely reverted to the initiating state $P$ by means of exactly one step
($\costBR=1$). 
Notably, both costs are independent from the length of the session, i.e.~from the number of forward interactions 
previously performed along the session and leading to the process $Q$. 

\item \textbf{Case (2).}\quad
We proceed by induction on $n$.\\ 
\noindent \emph{Base case} ($n=1$): the session is just initiated by means of 
rule \rulelabel{Fw(2)-Con}, thus it has the form $\singleSes{s}{P}Q$. The number of backward reductions 
to revert it is exactly one ($\costBR=1$), which is inferred by applying rule \rulelabel{Bw(2)-1}. Moreover, the dimension of
the memory stack of $\singleSes{s}{P}Q$ is exactly one ($\costMO=1$.). Thus, $\costBR=\costMO=1=n$. \\
\noindent
\emph{Inductive case}: let us consider a session $\singleSes{s}{m}Q$ with length $n$ and a 
session $\singleSes{s}{m'}Q'$ such that $\singleSes{s}{m}Q \fwredN{2} \singleSes{s}{m'}Q'$.
By definition of session length, $\singleSes{s}{m'}Q'$ has length $n+1$. 
Now, we have to prove that for this sessions $\costBR=\costMO=n+1$. 
The session $\singleSes{s}{m}Q$ can evolve to $\singleSes{s}{m'}Q'$ only 
by means of application of rule \rulelabel{Fw(2)-Mem}. Therefore, we have that 
$m' = Q\stackComp m$. 
By inductive hypothesis, the memory occupancy of $\singleSes{s}{m}Q$
is $n$, i.e. $|m|=n$. Thus, the memory occupancy of $\singleSes{s}{m'}Q'$ 
(i.e., the cost $\costMO$)  is $n+1$ because 
$|m'|=|Q\stackComp m|=1+|m|=n+1$.
Moreover, by inductive hypothesis we also have that the number of backward reduction steps
for reverting $\singleSes{s}{m}Q$ is $n$. By applying rule \rulelabel{Bw(2)-2} to  
$\singleSes{s}{m'}Q'$ we have 
$\singleSes{s}{m'}Q' \bwredN{2} \singleSes{s}{m}Q$. Therefore, the number of backward 
reduction steps for reverting $\singleSes{s}{m'}Q'$ (i.e., the cost $\costBR$) is $n+1$.

\item \textbf{Case (3).}\quad
We proceed by induction on $n$.\\ 
\noindent \emph{Base case} ($n=1$): the session is just initiated, 
thus it has the form $\singleSes{s}{P}Q$. The number of backward reductions 
to revert it is exactly one ($\costBR=1$), which is inferred by applying rule \rulelabel{Bw(3)-1}. Moreover, the dimension of
the memory stack of $\singleSes{s}{P}Q$ is exactly one ($\costMO=1$.). Thus, $\costBR=1$ and $\costMO=1=n$. \\
\noindent
\emph{Inductive case}: let us consider a session $\singleSes{s}{m}Q$ with length $n$ and a 
session $\singleSes{s}{m'}Q'$ such that $\singleSes{s}{m}Q \fwredN{3} \singleSes{s}{m'}Q'$.
By definition of session length, $\singleSes{s}{m'}Q'$ has length $n+1$. 
Now, we have to prove that for this sessions $\costBR=\costMO=n+1$. 
%
The session $\singleSes{s}{m}Q$ can evolve to $\singleSes{s}{m'}Q'$ only 
by means of application of rule \rulelabel{Fw(3)-Mem}. Therefore, we have that 
$m' = Q\stackComp m$. 
By inductive hypothesis, the memory occupancy of $\singleSes{s}{m}Q$
is $n$, i.e. $|m|=n$. Thus, the memory occupancy of $\singleSes{s}{m'}Q'$ 
(i.e., the cost $\costMO$)  is $n+1$ because 
$|m'|=|Q\stackComp m|=1+|m|=n+1$.
Now, since $|m'|>0$, the stack $m'$ must have a bottom element, say $P$, that is 
$m'=m''\stackComp P$ for some $m''$. Then, by applying rule \rulelabel{Bw(3)-3} to  
$\singleSes{s}{m'}Q'$ we have 
$\singleSes{s}{m''\stackComp P}Q' \bwredN{3} P$, i.e. the session is completely reversed in one backward 
reduction steps ($\costBR=1$).

\end{itemize}
\qed
\end{proof}

\lemmaApp{\ref{loopLemma}}
Let $P = \singleSes{s}{m}Q$ and $P' = \singleSes{s}{m'}Q'$ be two reachable processes in setting (i), with $i\in \{2,3\}$.
$P \fwredN{i} P'$ if and only if $P' \bwredN{i} P$.
\begin{proof}
We consider the case (2); the proof for case (3) proceeds similarly.  
Let us start with the proof for the \emph{if} part, which is by induction on the derivation of 
the forward reduction $\singleSes{s}{m}Q \fwredN{2} \singleSes{s}{m'}Q'$. 
According to the form of the source term of the reduction, which is a single session term,  
the only base case corresponds to the application of rule \rulelabel{Fw(2)-Mem}. In this case, we have that 
$m'=Q\stackComp m$.
By applying rule \rulelabel{Bw(2)-2} to $\singleSes{s}{m'}Q'$, we can directly 
conclude $\singleSes{s}{Q\stackComp m}Q' \bwredN{2} \singleSes{s}{m}Q$.
Concerning the inductive case, again, due to the the form of the source term of the reduction, there is only 
one rule that can be applied, that is \rulelabel{Fw-Str}. We have that 
$\singleSes{s}{m}Q \equiv \singleSes{s_1}{m_1}Q_1$,  
$\singleSes{s}{m'}Q' \equiv \singleSes{s_2}{m_2}Q_2$ 
and  
$\singleSes{s_1}{m_1}Q_1 \fwredN{2} \singleSes{s_2}{m_2}Q_2$. 
By induction $\singleSes{s_2}{m_2}Q_2 \bwredN{2} \singleSes{s_1}{m_1}Q_1$. 
Thus, we conclude by applying rule \rulelabel{Bw-Str}, since we directly get
$\singleSes{s}{m'}Q' \bwredN{2} \singleSes{s}{m}Q$.

Now, let us consider the proof for the \emph{only if} part, which is again 
by induction on the derivation of the forward reduction $\singleSes{s}{m'}Q'  \bwredN{2} \singleSes{s}{m}Q$.
In this case, the only base case corresponds to the rule \rulelabel{Bw(2)-2}. Therefore, $m'$ must be $Q\stackComp m$. 
By hypothesis, $\singleSes{s}{Q\stackComp m}Q' $ is a reachable process; hence, by Definition~\ref{def:reachable}, 
this term must be generated by a forward reduction 
$\singleSes{s}{m}Q \fwredN{2} \singleSes{s}{Q\stackComp m}Q'$, which allows us to conclude. 
As in the \emph{if} part, the inductive case corresponds to only rule \rulelabel{Bw-Str}. We have that 
$\singleSes{s}{m'}Q' \equiv \singleSes{s_1}{m_1}Q_1$,  
$\singleSes{s}{m}Q \equiv \singleSes{s_2}{m_2}Q_2$ 
and  
$\singleSes{s_1}{m_1}Q_1 \bwredN{2} \singleSes{s_2}{m_2}Q_2$. 
By induction $\singleSes{s_2}{m_2}Q_2 \fwredN{2} \singleSes{s_1}{m_1}Q_1$. 
Thus, we conclude by applying rule \rulelabel{Fw-Str}, since we directly get
$\singleSes{s}{m}Q \fwredN{2} \singleSes{s}{m'}Q'$. \qed
\end{proof}

\section{Proofs of Section~\ref{sec:rev_single_sessions_multi}}
\label{app:proofsSingleSessions_multi}

%
%
%
%
\theoremApp{\ref{prop:costs_multi}}
Let $n$ be the length of a session, the cost of reverting it is\\ 
\mbox{\ \ } -- case~(4): $\costBR=\costMO=1$;\\
\mbox{\ \ } -- case~(5): $\costBR=\costMO=n$;\\
\mbox{\ \ } -- case~(6): $\costBR=1$ and $\costMO=n$.
\begin{proof} We proceed separately with the three cases.
\begin{itemize}
\item \textbf{Case (4).}\quad
The proof is a trivial adaptation of the one of case (1) in Theorem~\ref{prop:costs}. Essentially, 
it is just needed to consider rules \rulelabel{Fw(4)-Con}, \rulelabel{Fw(4)-Mem} and the 
backward rule of case (4) in place of rules \rulelabel{Fw(1)-Con}, \rulelabel{Fw(1)-Mem}
and \rulelabel{Bw(1)}, respectively.

\item \textbf{Case (5).}\quad
This result can be proved by induction on $n$ by following the same steps of the proof of
case (2) in Theorem~\ref{prop:costs}. 
Indeed, the forward rules of cases (2) and (5) store the computation history in the same way, while 
the backward rules are exactly the same. Therefore, in the setting (5) both the backward reduction cost and the memory occupancy 
cost are the same of case (2), i.e. $\costBR=\costMO=n$.

\item \textbf{Case (6).}\quad
This result can be proved by induction on $n$ by following the same steps of the proof of
case (3) in Theorem~\ref{prop:costs}. 
Indeed, both the forward and backward rules of cases (3) and (6) are the same. 
Therefore, in the setting (6) both the backward reduction cost and the memory occupancy 
cost are the same of case (3), i.e. $\costBR=1$ and $\costMO=n$.
\end{itemize}
\end{proof}

\lemmaApp{\ref{loopLemma_multi}}
Let $P = \singleSes{s}{m}Q$ and $P' = \singleSes{s}{m'}Q'$ be two reachable processes in setting (i), with $i\in \{5,6\}$.
$P \fwredN{i} P'$ if and only if $P' \bwredN{i} P$.

\begin{proof}
This proof follows the same steps of the proof of Lemma~\ref{loopLemma}, as  
the forward rules of cases (2) and (5) store the computation history in the same way, 
and their backward rules are the same; and 
the forward and backward rules of cases (3) and (6) are the same. 
Let us consider the case (5); the proof for case (6) proceeds similarly.  
We start with the proof for the \emph{if} part, which is by induction on the derivation of 
the forward reduction $\singleSes{s}{m}Q \fwredN{5} \singleSes{s}{m'}Q'$.
According to the form of the source term of the reduction, which is a single session term,  
the only base case corresponds to the application of rule \rulelabel{Fw(5)-Mem}. In this case, we have that 
$m'=Q\stackComp m$.
By applying rule \rulelabel{Bw(5)-2} to $\singleSes{s}{m'}Q'$, we can directly 
conclude $\singleSes{s}{Q\stackComp m}Q' \bwredN{5} \singleSes{s}{m}Q$.
Concerning the inductive case, again, due to the the form of the source term of the reduction, there is only 
one rule that can be applied, that is \rulelabel{Fw-Str}. We have that 
$\singleSes{s}{m}Q \equiv \singleSes{s_1}{m_1}Q_1$,  
$\singleSes{s}{m'}Q' \equiv \singleSes{s_2}{m_2}Q_2$ 
and  
$\singleSes{s_1}{m_1}Q_1 \fwredN{5} \singleSes{s_2}{m_2}Q_2$. 
By induction $\singleSes{s_2}{m_2}Q_2 \bwredN{5} \singleSes{s_1}{m_1}Q_1$. 
Thus, we conclude by applying rule \rulelabel{Bw-Str}, since we directly get
$\singleSes{s}{m'}Q' \bwredN{5} \singleSes{s}{m}Q$.

Now, let us consider the proof for the \emph{only if} part, which is again 
by induction on the derivation of the forward reduction $\singleSes{s}{m'}Q'  \bwredN{5} \singleSes{s}{m}Q$.
In this case, the only base case corresponds to the rule \rulelabel{Bw(5)-2}. Therefore, $m'$ must be $Q\stackComp m$. 
By hypothesis, $\singleSes{s}{Q\stackComp m}Q' $ is a reachable process; hence, by Definition~\ref{def:reachable}, 
this term must be generated by a forward reduction 
$\singleSes{s}{m}Q \fwredN{5} \singleSes{s}{Q\stackComp m}Q'$, which allows us to conclude. 
As in the \emph{if} part, the inductive case corresponds to only rule \rulelabel{Bw-Str}. We have that 
$\singleSes{s}{m'}Q' \equiv \singleSes{s_1}{m_1}Q_1$,  
$\singleSes{s}{m}Q \equiv \singleSes{s_2}{m_2}Q_2$ 
and  
$\singleSes{s_1}{m_1}Q_1 \bwredN{5} \singleSes{s_2}{m_2}Q_2$. 
By induction $\singleSes{s_2}{m_2}Q_2 \fwredN{5} \singleSes{s_1}{m_1}Q_1$. 
Thus, we conclude by applying rule \rulelabel{Fw-Str}, since we directly get
$\singleSes{s}{m}Q \fwredN{5} \singleSes{s}{m'}Q'$. \qed
\end{proof}

\else
\fi

\end{document}